\documentclass[journal]{IEEEtran}

\usepackage{amsmath}
\allowdisplaybreaks[4]
\usepackage{amsthm}
\usepackage{amssymb}
\usepackage{verbatim}
\usepackage{algorithm}
\usepackage{algorithmic}
\usepackage{bbding}
\usepackage{graphicx}
\usepackage{multirow}
\usepackage{multicol}
\usepackage{booktabs}
\usepackage{makecell}
\usepackage{yhmath}
\usepackage{paralist}
\usepackage{color}
\usepackage{xr-hyper}
\usepackage[colorlinks,
 linkcolor=black,
 urlcolor=black,
 anchorcolor=black,
 citecolor=black]{hyperref}
\usepackage{cite}
\usepackage{enumitem}
\usepackage{xcolor,colortbl}
\usepackage{dsfont}
\usepackage{float}
\usepackage{epstopdf}
\usepackage{subfigure}
\usepackage{lipsum}
\newtheorem{lemma}{Lemma}
\newtheorem{definition}{Definition}

\newtheorem{proposition}{Proposition}
\newcommand\scalemath[2]{\scalebox{#1}{\mbox{\ensuremath{\displaystyle #2}}}}

\begin{document}
\bstctlcite{BSTcontrol}

\title{Joint Beam Direction Control and Radio Resource Allocation in Dynamic Multi-beam LEO Satellite Networks}

\author{{Shuo~Yuan,~\IEEEmembership{Member,~IEEE},
			Yaohua~Sun,
			Mugen~Peng,~\IEEEmembership{Fellow,~IEEE},
			and~Renzhi~Yuan,~\IEEEmembership{Member,~IEEE}}
	\thanks{Copyright (c) 2015 IEEE. Personal use of this material is permitted. However, permission to use this material for any other purposes must be obtained from the IEEE by sending a request to pubs-permissions@ieee.org.}
	\thanks{This work was supported in part by the National Key R$\&$D Program of China under Grant 2022YFB2902600, in part by the National Natural Science Foundation of China under Grants 62371071, and in part by the Young Elite Scientists Sponsorship Program by the China Association for Science and Technology under Grant 2021QNRC001. (\emph{Corresponding author: Mugen Peng.})}
	\thanks{The authors are with the State Key Laboratory of Networking and Switching Technology, Beijing University of Posts and Telecommunications, Beijing 100876, China (e-mail: yuanshuo@bupt.edu.cn; sunyaohua@bupt.edu.cn; pmg@bupt.edu.cn; renzhi.yuan@bupt.edu.cn).}
}

\maketitle

\begin{abstract}
	Multi-beam low earth orbit (LEO) satellites are emerging as key components in beyond 5G and 6G to provide global coverage and high data rate.
	To fully unleash the potential of LEO satellite communication, resource management plays a key role.
	However, the uneven distribution of users, the coupling of multi-dimensional resources, complex inter-beam interference, and time-varying network topologies all impose significant challenges on effective communication resource management.
	In this paper, we study the joint optimization of beam direction and the allocation of spectrum, time, and power resource in a dynamic multi-beam LEO satellite network.
	The objective is to improve long-term user sum data rate while taking user fairness into account.
	Since the concerned resource management problem is mixed-integer non-convex programming, the problem is decomposed into three subproblems, namely beam direction control and time slot allocation, user subchannel assignment, and beam power allocation.
	Then, these subproblems are solved iteratively by leveraging matching with externalities and successive convex approximation, and the proposed algorithms are analyzed in terms of stability, convergence, and complexity.
	Extensive simulations are conducted, and the results demonstrate that our proposal can improve the number of served users by up to two times and the sum user data rate by up to 68\%, compared to baseline schemes.
\end{abstract}


\IEEEpeerreviewmaketitle

\section{Introduction}

\label{sec:intro}
\IEEEPARstart{S}{atellite} communication has been identified as one of the key elements in beyond 5G and 6G to complement terrestrial networks to achieve seamless coverage and cost-efficient data services \cite{chen2023trendschallenges}.
Low earth orbit (LEO) satellites, operating at altitudes of $500-1500\,km$, are particularly attractive due to shorter propagation latency, higher data rates, and lower launch costs compared to geostationary earth orbit (GEO) satellites.
Currently, several commercial LEO communication satellite constellations have emerged, such as Starlink and OneWeb \cite{liu2021leosatellite}.

In highly dense satellite constellations, effectively scheduling communication resources is crucial for improving network capacity and addressing the uneven distribution of user traffic; however, this still presents many challenges to overcome \cite{lin2020integratedgsatellite}.
Firstly, severe inter-beam interference can exist when full-frequency reuse among beams is utilized.
Secondly, with the development of software-defined communication payload, an LEO satellite can flexibly allocate spectrum, power, time slots, and beams among users, and meanwhile beam direction can also be adjusted flexibly with the use of advanced antenna technology.
In general, from the perspective of network capacity, all these factors are coupled, and their joint consideration makes resource management highly complicated.
Thirdly, the dynamic nature of the LEO satellite network topology results in time-varying states of inter-beam interference and satellite-terrestrial link conditions, requiring resource management to adapt to such dynamics \cite{yuan2023softwaredefined,zhu2023timingadvance,yuan2023jointnetwork}.

\subsection{Related Works}

\emph{1) Radio Resource Management:} Recently, flexible resource allocation in satellite communication has received remarkable attention \cite{li2020spectrumoptimization,gu2022dynamiccooperative,kawamoto2020flexibleresource,choi2005optimumpower,lin2019jointbeamforming,abdu2021flexibleresource,zhou2021machinelearningbased,yuan2022adaptingdynamic}.
In \cite{li2020spectrumoptimization}, a market-based spectrum optimization problem for a multi-beam GEO satellite system is proposed to enhance spectrum efficiency, and the spectrum pricing model is built by utilizing the Hotelling model.
In \cite{gu2022dynamiccooperative}, the authors investigate cooperative spectrum sharing between LEO and GEO satellites to maximize the sum data rate of LEO satellite users under the constraints of beam transmission power and signal-to-interference ratio of GEO satellite users.
Assuming a GEO satellite equipped with a digital channelizer, the authors of \cite{kawamoto2020flexibleresource} study frequency resource allocation among beams to alleviate inter-beam interference and maximize network throughput.
In \cite{choi2005optimumpower}, a power allocation and multi-beam scheduling approach is proposed to match limited power resources with the heterogeneous traffic demands of unevenly distributed users while investigating the tradeoff between total capacity and user fairness.
The authors in \cite{lin2019jointbeamforming} investigate the problem of joint power allocation and beamforming for a GEO satellite-based communication system to improve the sum data rate, and propose a cooperative user pairing and an iterative penalty function-based beamforming schemes.

In addition, both subchannel assignment and power allocation are optimized in \cite{abdu2021flexibleresource} for a multi-beam GEO satellite communication system to meet user traffic demand while using minimum transmission power and bandwidth.
The authors of \cite{zhou2021machinelearningbased} formulate a joint power resource allocation and data transmission scheduling problem for a satellite-assisted internet of remote things network with the aim of maximizing the sum data rate, and a model-free reinforcement learning framework is designed to accommodate highly dynamic satellite communication environments.
In \cite{yuan2022adaptingdynamic}, a radio resource scheduling problem is investigated, intending to maximize the number of served users and the sum data rate in an integrated satellite-terrestrial network and a meta-critic learning based solution is proposed.

\emph{2) Beam Direction Control:} On the other hand, as the number of satellites increases, it will be possible for multiple satellites on different orbits to serve the same ground area \cite{su2019broadbandleo, xv2023jointbeam}.
At this time, satellite beams need to be pointed at proper locations to alleviate inter-beam interference and better serve unevenly distributed users with various traffic demands.
In \cite{yin2021beampointing}, a satellite beam direction optimization problem is formulated for a scenario with multiple non-GEO satellites sharing the same spectrum and a genetic algorithm based solution is proposed.
To minimize the number of activated beams while holding all users close to beam centers, the authors of \cite{bui2022jointbeam} investigate a joint beam direction and load balancing optimization problem with a fixed beam transmission power budget.
Aiming to maximize the total throughput in the concerned geographic area, the authors of \cite{takahashi2022dbfbasedfusion} investigate a beam fusion control optimization problem, where the control parameters include transmission power, beam gain, and beam direction.
In \cite{wang2022jointoptimization}, the authors focus on a beam hopping based and non-orthogonal-multiple-access enhanced multi-beam GEO satellite system, and beam transmission power, beam scheduling as well as time slot allocation among users are jointly addressed to minimize the gap between user demands and offered data rate.
In \cite{hu2020dynamicbeam}, a beam hopping optimization problem is formulated to guarantee the fairness of beam services and meanwhile maximize service quality, and deep reinforcement learning is utilized to tackle the problem.

\subsection{Motivations and Contributions}

\begin{table}[]
	\small
	\caption{Related Works}
	\centering
	\label{tab:relatedWorks}
	\scalebox{0.78}{
		\begin{tabular}{|c|cccccc|}
			\hline
			\multirow{2}{*}{Literature}                        & \multicolumn{6}{c|}{Target Scenario}                                                                                                                                                      \\ \cline{2-7} & \multicolumn{1}{c|}{Orbit Type} & \multicolumn{1}{c|}{\begin{tabular}[c]{@{}c@{}}Multiple\\ Satellites\end{tabular}} & \multicolumn{1}{c|}{\begin{tabular}[c]{@{}c@{}}Dynamic\\ Topology\end{tabular}} & \multicolumn{1}{c|}{\begin{tabular}[c]{@{}c@{}}Beam\\ Direction\end{tabular}} & \multicolumn{1}{c|}{Carrier}     & Power       \\ \hline
			\cite{li2020spectrumoptimization}                  & \multicolumn{1}{c|}{GEO}             & \multicolumn{1}{c|}{ }          & \multicolumn{1}{c|}{ }          & \multicolumn{1}{c|}{ }          & \multicolumn{1}{c|}{\Checkmark} &            \\ \hline
			\cite{gu2022dynamiccooperative}                    & \multicolumn{1}{c|}{LEO}             & \multicolumn{1}{c|}{\Checkmark} & \multicolumn{1}{c|}{ }          & \multicolumn{1}{c|}{ }          & \multicolumn{1}{c|}{\Checkmark} & \Checkmark \\ \hline
			\cite{kawamoto2020flexibleresource}                & \multicolumn{1}{c|}{GEO/LEO}         & \multicolumn{1}{c|}{ }          & \multicolumn{1}{c|}{ }          & \multicolumn{1}{c|}{ }          & \multicolumn{1}{c|}{\Checkmark} &            \\ \hline
			\cite{choi2005optimumpower}                        & \multicolumn{1}{c|}{GEO/LEO}         & \multicolumn{1}{c|}{ }          & \multicolumn{1}{c|}{ }          & \multicolumn{1}{c|}{ }          & \multicolumn{1}{c|}{ }          & \Checkmark \\ \hline
			\cite{lin2019jointbeamforming}                     & \multicolumn{1}{c|}{GEO}             & \multicolumn{1}{c|}{ }          & \multicolumn{1}{c|}{ }          & \multicolumn{1}{c|}{ }          & \multicolumn{1}{c|}{\Checkmark} & \Checkmark \\ \hline
			\cite{abdu2021flexibleresource}                    & \multicolumn{1}{c|}{GEO}             & \multicolumn{1}{c|}{ }          & \multicolumn{1}{c|}{ }          & \multicolumn{1}{c|}{ }          & \multicolumn{1}{c|}{\Checkmark} & \Checkmark \\ \hline
			\cite{zhou2021machinelearningbased}                & \multicolumn{1}{c|}{LEO}             & \multicolumn{1}{c|}{\Checkmark} & \multicolumn{1}{c|}{\Checkmark} & \multicolumn{1}{c|}{ }          & \multicolumn{1}{c|}{ }          & \Checkmark \\ \hline
			\cite{yuan2022adaptingdynamic}                     & \multicolumn{1}{c|}{LEO}             & \multicolumn{1}{c|}{\Checkmark} & \multicolumn{1}{c|}{\Checkmark} & \multicolumn{1}{c|}{ }          & \multicolumn{1}{c|}{\Checkmark} &            \\ \hline
			\cite{yin2021beampointing}                         & \multicolumn{1}{c|}{non-GEO}         & \multicolumn{1}{c|}{\Checkmark} & \multicolumn{1}{c|}{\Checkmark} & \multicolumn{1}{c|}{\Checkmark} & \multicolumn{1}{c|}{ }          &            \\ \hline
			\cite{bui2022jointbeam}                            & \multicolumn{1}{c|}{non-GEO}         & \multicolumn{1}{c|}{ }          & \multicolumn{1}{c|}{ }          & \multicolumn{1}{c|}{\Checkmark} & \multicolumn{1}{c|}{ }          &            \\ \hline
			\cite{takahashi2022dbfbasedfusion}                 & \multicolumn{1}{c|}{GEO/LEO}         & \multicolumn{1}{c|}{ }          & \multicolumn{1}{c|}{ }          & \multicolumn{1}{c|}{\Checkmark} & \multicolumn{1}{c|}{}           & \Checkmark \\ \hline
			\cite{wang2022jointoptimization,hu2020dynamicbeam} & \multicolumn{1}{c|}{GEO}             & \multicolumn{1}{c|}{ }          & \multicolumn{1}{c|}{\Checkmark} & \multicolumn{1}{c|}{\Checkmark} & \multicolumn{1}{c|}{ }          & \Checkmark \\ \hline
			This Work                                          & \multicolumn{1}{c|}{LEO}             & \multicolumn{1}{c|}{\Checkmark} & \multicolumn{1}{c|}{\Checkmark} & \multicolumn{1}{c|}{\Checkmark} & \multicolumn{1}{c|}{\Checkmark} & \Checkmark \\ \hline
		\end{tabular}}
\end{table}

In Table \ref{tab:relatedWorks}, existing literature about beam direction control and radio resource allocation for satellite communication is summarized.
Based on the literature review, it can be found that most existing works focus on flexible radio resource allocation with fixed beam direction for either one satellite, e.g., \cite{li2020spectrumoptimization,kawamoto2020flexibleresource,choi2005optimumpower,lin2019jointbeamforming}, or multiple satellites, e.g., \cite{gu2022dynamiccooperative,abdu2021flexibleresource,zhou2021machinelearningbased,yuan2022adaptingdynamic}.
In terms of beam direction control, some works \cite{yin2021beampointing, bui2022jointbeam} solely focus on the geographic distribution of users, neglecting the impact of the dynamic network topology introduced by the mobility of LEO satellites.
Although some works \cite{takahashi2022dbfbasedfusion,wang2022jointoptimization,hu2020dynamicbeam} are dedicated to jointly optimizing power allocation and beam direction, only scenarios with a single satellite are considered.
In summary, there are few works addressing the joint optimization of beam direction and multi-dimensional radio resource allocation for a target ground area in multi-beam multi-LEO-satellite networks with differentiated data rate demands of users and dynamic network topologies.

Further, we conclude the research motivations of this paper as follows:
\begin{itemize}
	\item
	      The joint optimization of multi-dimensional radio resource allocation and beam direction in multi-beam LEO satellite networks has not been thoroughly studied.
	      In addition to adapting to the uneven geographical distribution of users, data rate demands of individual users should also be considered to satisfy basic communication demands.
	\item
	      Most of the existing works on resource allocation for LEO satellite networks have not taken the impacts of dynamic satellite-terrestrial topology into account, which leads to time-varying inter-beam interference and communication link state.
	\item
	      The allocation of spectrum, transmission power, and time slots is tightly coupled with beam direction control, and their joint optimization is usually NP-hard, which is non-trivial to obtain an optimal solution within polynomial time.
\end{itemize}

With these motivations, in this work, a resource management problem is formulated for the downlink scenario in a dynamic multi-beam multi-satellite network, which takes into account both beam direction control and multi-dimensional radio resource allocation, and we aim to enhance long-term sum user data rate while improving user fairness.
To tackle the problem, we decouple the primal problem and provide a joint solution framework where beam direction control and time slot allocation, subchannel assignment, and power allocation are executed iteratively.
More specifically, we find that the beam direction control and time slot allocation subproblem can be modeled as a two sided many-to-one matching with externalities between beam center coordinate-time-slot units and beams.
In addition, the subchannel assignment subproblem can be modeled as a many-to-one matching with externalities between subchannel-time-slot units and users.
The power allocation subproblem is handled using successive convex approximation (SCA), which transforms it into a convex problem.
To solve the first two subproblems, we design two modified swap/negotiation operation based matching algorithms, while the third subproblem is solved using a general convex optimization tool.
The main contributions of this paper are concluded as follows.
\begin{itemize}
	\item
	      A joint beam direction control and multi-dimensional resource allocation problem is formulated for the downlink scenario in a time-slotted dynamic multi-beam multi-satellite network. The aim is to enhance the long-term sum user data rate of a target area while considering user fairness.
	      The formulated problem is mixed-integer non-convex programming, which is proved to be NP-hard.
	\item
	      In order to solve this joint optimization problem, we first decouple it into three subproblems, including beam direction control and time slot allocation, subchannel assignment, and power allocation.
	      The first two subproblems are modeled using many-to-one matching theory with externalities, and then two modified swap/negotiation operation based matching algorithms are proposed. Their stability, convergence, and complexity are analyzed.
	      Afterward, we invoke SCA method to iteratively update beam power allocation by solving an approximated convex problem.
	\item
	      Extensive simulations are conducted, demonstrating that the proposed beam direction control and radio resource allocation algorithms outperform the benchmarks in terms of long-term sum user data rate, the number of served users, and user fairness.
	      We also investigate the impacts of the number of beams configured at each satellite, the number of available subchannels per beam, and the maximum number of subchannels that can be allocated to each user.
	      In addition, the generality and robustness of our proposal under different user distributions are demonstrated.
\end{itemize}

The rest of the paper is organized as follows.
In Section \ref{sec:systModel}, we present the system model of the multi-beam multi-satellite downlink network and formulate the long-term data rate optimization problem.
Then, by decomposing the primal problem, an iterative solution framework for joint beam direction control and radio resource allocation is proposed in Section \ref{sec:solutionDesign}.
In Section \ref{sec:numericalRes}, simulation results along with analysis are provided.
Finally, conclusions are drawn in Section \ref{sec:conclusion}.

\begin{figure}[t]
	\centering
	\includegraphics[width=0.4\textwidth]{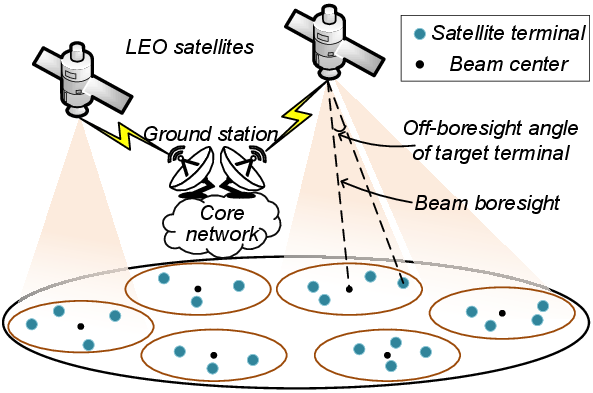}
	\caption{An illustration of the studied satellite network scenario.}
	\label{fig:NETtopology}
\end{figure}

\section{System Model and Problem Formulation}
\label{sec:systModel}

As shown in Fig. \ref{fig:NETtopology}, we consider a multi-beam multi-satellite downlink network, which consists of a set of LEO satellites, ground stations, and multiple user terminals.
Each LEO satellite is equipped with an array-fed reflector antenna, consisting of a feed array with $F$ feeds to generate $L$ beams.
Generally, the array-fed reflector antenna is characterized by two structural paradigms: single feed per beam and multiple feeds per beam \cite{chu2021robustdesign}.
This paper follows the approach of a single feed per beam, as seen in the modern LEO satellite antenna structure presented in \cite{gu2022dynamiccooperative, khan2023ratesplitting}, meaning that only one feed is required to generate one beam.
All user terminals are furnished with a single antenna and are located within a circular ground area $S$ with a radius of $R_s$ on the surface of the earth.
In addition, the satellite ground station integrates network management functions, encompassing global status collection and strategy generation \cite{sheng2017flexiblereconfigurable}.

We denote the set of user terminals and the set of LEO satellites as ${\mathcal N} = \{ 1, ..., n, ..., N\}$ and ${\mathcal M} = \{ 1, ...,m, ..., M\}$, respectively.
Each satellite $m$ has a set of beams, denoted by ${\mathcal L}_m$, and the set of all the satellite beams is denoted as ${\mathcal Q} = \{ {{\mathcal L}_1} \cup ... \cup {{\mathcal L}_M} \}$.
A full-frequency reuse scheme, where all the beams use the same frequency band, is adopted to raise resource utilization.
The user terminals associated with the same beam are served based on orthogonal frequency division multiple access (OFDMA), and the frequency band occupied by each beam is divided into $K$ subchannels, whose set is denoted by ${\mathcal K} = \{ 1, ...,k, ..., K\}$.
In addition, system operation time is divided into discrete time slots, whose set is denoted by ${\mathcal T} = \{ 1, ...,t ..., T\}$.
The network topology, similar to \cite{yuan2022adaptingdynamic} and \cite{cao2023edgeassistedmultilayer}, is assumed to be quasi-static within each time slot of length $\tau$ but exhibits variation between time slots owing to the mobility of LEO satellites.

\subsection{Channel Model}
\label{sec:channelModel}

For user terminal $n$ served by beam $q\in\mathcal Q$ whose beam center position is with Earth-Centered, Earth-Fixed (ECEF) coordinate $c$, user-satellite channel state $h_{q,c,k,n}^t$ on subchannel $k$ at time slot $t$ is modeled as
\begin{equation}
	\label{eq:channelGain}
	\begin{aligned}
		h_{q,c,k,n}^t = {G_{q,c,k,n}^{{\rm{tx}},t} \cdot G_n^{{\rm{rx}}} \cdot G_{q,k,n}^{t}}, \\
		{\forall q \in {\mathcal Q}, k \in {\mathcal K}, n \in {\mathcal N}, t \in {\mathcal T} },
	\end{aligned}
\end{equation}
where $G_{q,c,k,n}^{{\rm{tx}},t}$, $G_n^{{\rm{rx}}}$, and $G_{q,k,n}^{t}$ are the transmission antenna gain, the user receiving antenna gain, and the channel gain, respectively.
Based on the antenna structure, the transmission antenna gain between beam $q$ with beam center position $c$ and user terminal $n$ at time slot $t$, $G_{q,c,k,n}^{{\rm{tx}},t}$, is given by \cite{gu2022dynamiccooperative, khan2023ratesplitting, huang2020energyefficient}
\begin{equation}
	\label{eq:channelTransGain}
	G_{q,c,k,n}^{{\rm{tx}},t} = G_q^{{\rm{peak}}}{\left({\frac{{{J_1}\left( \mu \right)}}{{2\mu }} + 36\frac{{{J_3}\left( \mu \right)}}{{{\mu ^3}}}}\right)^2},
\end{equation}
where $G_q^{{\rm{peak}}} = \eta {(\frac{{\pi D {f}}}{\epsilon})^2}$ \cite{ippolito2008satellitecommunications} is the peak transmission antenna gain of beam $q$.
$\mu =2.07123\sin(\theta _{q,c,n}^{t})/\sin(\theta _{q,3\mathrm{dB}}^{t})$, where $\theta _{q,3{\rm{dB}}}^t$ denotes the $3 \rm{dB}$ angle of beam $q$, and ${\theta _{q,c,n}^{t}}$ signifies the off-boresight angle of user terminal $n$ relative to the boresight of beam $q$ with beam center position $c$.
The terms ${J_1}$ and ${J_3}$ represent the first and third order of first kind of Bessel functions, respectively.
In addition, $\eta$ denotes the aperture efficiency of the satellite antenna, $D$ is the diameter of the satellite antenna aperture, ${f}$ is the carrier frequency, and $\epsilon$ is the speed of light.

The off-boresight angle ${\theta _{q,c,n}^{t}}$, as shown in Fig. \ref{fig:NETtopology}, is contingent on the ECEF coordinates of the satellite to which beam $q$ belongs, the beam center position of beam $q$ and user terminal $n$.
Specifically, the ECEF coordinate of the satellite is denoted by $\omega_{q}^{t}$ and can be obtained from the precise ephemeris.
Denote the distance between the satellite to which beam $q$ belongs and user terminal $n$, the distance between the satellite and the beam center position, and the distance between user terminal $n$ and beam center position of beam $q$ as $d_{q,n}^t$, ${d_{q,c}^t}$, and $d_{q,c,n}^t$, respectively.
Then, off-boresight angle ${\theta _{q,c,n}^{t}}$ can be given by the law of cosines as
\begin{equation}
	\label{eq:offBeamAngle}
	\theta _{q,c,n}^t = \arccos \left(\frac{{{\left({d_{q,n}^t}\right)^2} + {\left({{d_{q,c}^t}}\right)^2} - {\left({d_{{q},c,n}^t}\right)^2}}}{{2d_{q,n}^t{d_{q,c}^t}}}\right).
\end{equation}
According to the configured beam center position, we can obtain
\begin{equation}
	\label{eq:distBeamCent}
	\begin{aligned}
		{\left({{d_{q,c}^t}}\right)^2} = {{\left\| {c - \omega _q^t} \right\|}_2^2}, \quad {\left({d_{{q},c,n}^t}\right)^2} = {{\left\| {c - {\omega _n}} \right\|}_2^2},
	\end{aligned}
\end{equation}
where $\left\| {a - b} \right\|_2$ is used to find the Euclidean distance between coordinate $a$ and coordinate $b$, and ${\omega _n}$ is the ECEF coordinate of user terminal $n$.

The user-satellite channel gain is modeled using a widely adopted channel fading model, which takes into account free-space path loss, atmospheric fading, and Rician small-scale fading \cite{yuan2022adaptingdynamic}.
The channel gain between the satellite to which beam $q$ belongs and user terminal $n$ on subchannel $k$ at time slot $t$, denoted as $G_{q,k,n}^{t}$, is modeled as
\begin{equation}
	\label{eq:channelFadingGain}
	G_{q,k,n}^{t} = {\left( {\frac{\epsilon}{{4\pi d_{q,n}^t{f}}}} \right)^2} A^{-1}\left( {d_{q,n}^t} \right) \rho _1 ,
\end{equation}
where $\rho _1 $ is the Rician fading factor, and the atmospheric loss $A(d_{q,n}^{t})$ is expressed as \cite{ippolito2008satellitecommunications}
\begin{equation}
	A(d_{q,n}^{t})=10^{\frac{d_{q,n}^{t}( 4.343\rho _2+\rho _3 )}{10H_q}},
\end{equation}
with $\rho _2$ and $\rho _3$ being the attenuation factors of the clouds and rain, respectively, and $H_q$ is the height of the satellite to which beam $q$ belongs.

\subsection{Problem Formulation}
\label{sec:problemForm}

Considering the practical constraints on the operational complexity of satellite antennas, beams are configured with a limited number of potential center coordinates on the surface of Earth.
We denote the set of optional beam center coordinates within the target ground area $S$ by ${\mathcal C}=\{1,...,c,...,C\}$.
Binary variable $e_{c,q}^t$ is used to indicate whether beam center candidate coordinate $c$ is selected for beam $q$ at time slot $t$, and we have
\begin{equation}
	\label{eq:beamCenterIndicator}
	e_{c,q}^t =
	\begin{cases}
		1, & \text{if coordinate } c \text{ is configured}   \\
		   & \text{for beam } {q} \text{ at time slot } {t}, \\
		0, & \text {otherwise. }
	\end{cases}
\end{equation}
In addition, we define ${a_{q,n}^t}$ as a binary user association indicator, where ${a_{q,n}^t}=1$ indicates that user terminal $n$ is served by beam $q$ at time slot $t$, otherwise, ${a_{q,n}^t}=0$.
In this paper, each user terminal is assumed to access the beam with the best received signal strength.
For example, the beam serving user terminal $n$ at time slot $t$ is
\begin{equation}
	q^{\ast} = \arg \max_{q\in \mathcal{Q}} \sum_{c\in \mathcal{C}}{e_{c,q}^{t}h_{q,c,k,n}^t p^{\prime}},
\end{equation}
where $p^{\prime}$ is the downlink reference signal transmission power.
Then, we have $a_{q^{\ast},n}^{t} = 1$ with $\sum_{q\in \mathcal{Q}}{a_{q,n}^{t}=1}$.
Moreover, user subchannel assignment is represented by a binary variable ${b_{k,n}^t}$, which is defined as
\begin{equation}
	\label{eq:subchannelAssignVar}
	{b_{k,n}^t} =
	\begin{cases}
		1, & \text{if subchannel } k \text{ is assigned to }      \\
		   & \text{user terminal } {n} \text{ at time slot } {t}, \\
		0, & \text {otherwise. }
	\end{cases}
\end{equation}

The signal-to-interference-plus-noise ratio (SINR) of user terminal $n$ served by beam $q$ on subchannel $k$ is given by
\begin{equation}
	\label{eq:subchannelSINR}
	\gamma _{q,c,k,n}^t = \frac{{h_{q,c,k,n}^tp_{q,k}^t}}{{{I_{{\rm{inter}}}} + {\sigma _{n}^2}}},
\end{equation}
where $p_{q,k}^{t}=p_{m,q}^{t}/K$ is the beam transmission power on subchannel $k$ with $p_{m,q}^{t}$ being the power allocated to beam $q$ by satellite $m$ at time slot $t$ for $q \in \mathcal{L}_m$, ${\sigma _{n}^2}$ is the additive white Gaussian noise, and ${I_{{\rm{inter}}}}$ is inter-beam interference, which is expressed as
\begin{small}
	\begin{equation}
		\label{eq:subchannelInterInf}
		{I_{{\rm{inter}}}} =  \sum\limits_{q' \in {\mathcal Q} \backslash q} \sum\limits_{n' \in {{\mathcal N}\backslash n}} \sum\limits_{c' \in {\mathcal C} \backslash c} { {a_{q',n'}^t b_{k,n'}^t e_{c',q'}^t}{h_{q',c',k,n}^tp_{q',k}^t} }.
	\end{equation}
\end{small}

Then, when user terminal $n$ is served by beam $q$ on subchannel $k$ at time slot $t$, the data rate is given by
\begin{equation}
	\label{eq:subchannelDataRate}
	R_{q,c,k,n}^t = {{B_{q,k}}}{\log _2}\left( {1 + \gamma _{q,c,k,n}^t} \right),
\end{equation}
where ${B_{q,k}} = {B_{q}} / K$ is the bandwidth of subchannel $k$ with ${B_{q}}$ being the bandwidth of beam $q$.
As a result, the general expression of the data rate of user terminal $n$ during time slot $t$ is
\begin{equation}
	\label{eq:terminalDataRate}
	R_{q,n}^t = \sum\limits_{k \in {{\mathcal K}}} \sum\limits_{c \in {{\mathcal C}}} {a_{q,n}^tb_{k,n}^t}{e_{c,q}^t}R_{q,c,k,n}^t.
\end{equation}

To achieve a trade-off between network throughput and user fairness, we involve $\alpha$-proportional utility function $U^{\alpha}(\cdot)$ \cite{zhao2017spectrumallocation} to formulate the objective of resource management, which is defined as
\begin{equation}
	\label{eq:alphaPropUF}
	U^{\alpha}(x)=
	\begin{cases}
		 & \left( 1-\alpha \right) ^{-1}x^{1-\alpha}, 0 \le \alpha <1, \\
		 & \log \left( x \right), \alpha = 1.                          \\
	\end{cases}
\end{equation}
and the $\alpha$-utility of user terminal $n$ during $T$ time slots is $U^{\alpha}( \sum_{t\in \mathcal{T}}{\sum_{q\in \mathcal{Q}}{R_{q,n}^{t}}}) $.

This paper aims to jointly optimize beam direction $\boldsymbol{\rm{e}} = \{ e_{c,q}^t , c \in {{\mathcal C}}, q \in {\mathcal Q}, t \in {{\mathcal T}} \}$, user subchannel assignment $\boldsymbol{\rm{b}} = \{b_{k,n}^t, k \in {{\mathcal K}}, n \in {\mathcal N}, t \in {{\mathcal T}}\}$, and beam power allocation $\boldsymbol{\rm{p}} = \{ {p_{m,q}^t}, m \in {{\mathcal M}}, q \in {{\mathcal Q}},t \in {{\mathcal T}}\}$ to maximize the sum of $\alpha$-utility of all user terminals.
The optimization problem is formally given as follows.
\begin{subequations}
	\label{eq:primalProblem}
	\begin{align}
		\begin{split}
			\label{eq:ppOJ}
			&\mathop {\max }\limits_{\boldsymbol{\rm{b,e,p}}} \quad
			\sum_{n\in \mathcal{N}}{U^{\alpha}( \sum_{t\in \mathcal{T}}{\sum_{q\in \mathcal{Q}}{R_{q,n}^{t}}})}
		\end{split}
		\\
		\begin{split}
			\emph{s.t.}
			\label{eq:constSubcAssi}
			&\sum\limits_{n \in {\mathcal N}} {a_{q,n}^tb_{k,n}^t} \le 1,\forall q \in {\mathcal Q},k \in {\mathcal K},t \in {\mathcal T},
		\end{split}
		\\
		\begin{split}
			\label{eq:constSubcAssiMax}
			&\sum\limits_{k \in {\mathcal K}} {a_{q,n}^tb_{k,n}^t} \le K^{thr},\forall q \in {\mathcal Q},n \in {\mathcal N},t \in {\mathcal T},
		\end{split}
		\\
		\begin{split}
			\label{eq:constBeamMaxSubs}
			&\sum\limits_{n \in {\mathcal N}} {\sum\limits_{k \in {\mathcal K}} {a_{q,n}^tb_{k,n}^t} } \le K,\forall q \in {{\mathcal Q}},t \in {{\mathcal T}},
		\end{split}
		\\
		\begin{split}
			\label{eq:constBeamCoordSelection}
			&\sum\limits_{c \in {{\mathcal C}}} {e_{c,q}^t} \le {1},\forall q \in {\mathcal Q},t \in {\mathcal T},
		\end{split}
		\\
		\begin{split}
			\label{eq:constCoordBEselected}
			&\sum\limits_{q \in {{\mathcal Q}}} {e_{c,q}^t} \le 1,\forall c \in {\mathcal C},t \in {\mathcal T},
		\end{split}
		\\
		\begin{split}
			\label{eq:constMaxPowerLmtBeam}
			&{p_{m,q}^t} \le P_q^{\max },\forall q \in {{\mathcal Q}},t \in {{\mathcal T}},
		\end{split}
		\\
		\begin{split}
			\label{eq:constMaxPowerLmt}
			&\sum\limits_{q \in {{{\mathcal L}}_m}} {p_{m,q}^t} \le P_m^{\max },\forall m \in {{\mathcal M}},t \in {{\mathcal T}},
		\end{split}
		\\
		\begin{split}
			\label{eq:constMinPower}
			&0 \le p_q^t, \forall q \in {{\mathcal Q}}, t \in {{\mathcal T}},
		\end{split}
		\\
		\begin{split}
			\label{eq:constMinEleVAngle}
			& \theta _{q,n,{\rm{ele}}}^t \ge a_{q,n}^t{\theta _0}, \forall q \in {{\mathcal Q}},n \in {{\mathcal N}},t \in {{\mathcal T}},
		\end{split}
		\\
		\begin{split}
			\label{eq:constMinSINR}
			\gamma _{q,c,k,n}^t \ge b_{k,n}^t{\gamma _0}, \forall q \in {{\mathcal Q}},k \in {{\mathcal K}},n \in {{\mathcal N}},t \in {{\mathcal T}},
		\end{split}
		\\
		\begin{split}
			\label{eq:constBinaryB}
			&b_{k,n}^t \in \{ 0, 1\}, \forall k \in {{\mathcal K}}, n \in {\mathcal N}, t \in {{\mathcal T}},
		\end{split}
		\\
		\begin{split}
			\label{eq:constBinaryE}
			&e_{c,q}^t \in \{ 0, 1\}, \forall c \in {{\mathcal C}}, q \in {\mathcal Q}, t \in {{\mathcal T}}.
		\end{split}
	\end{align}
\end{subequations}
where $P_q^{\max }$ and $P_m^{\max}$ denote the transmission power limitation of beam $q$ and satellite $m$, respectively.
$\theta _{q,n,{\rm{ele}}}^t = {\arccos ( {\frac{{{{\left( {d_{q,n}^t} \right)}^2} + {R^2} - {{\left( {H_q + R} \right)}^2}}}{{2d_{q,n}^tR}}} ) - \frac{\pi }{2}}$ denotes the elevation angle between the satellite to which beam $q$ belongs and user terminal $n$ at time slot $t$ and ${\theta _0}$ represents the minimum required elevation angle.
${\gamma _0}$ is the minimum required SINR to satisfy basic communication demands.
The orthogonal subchannel assignment in each beam is indicated by constraint \eqref{eq:constSubcAssi}.
Constraint \eqref{eq:constSubcAssiMax} limits the maximum number of subchannels, i.e., $K^{thr}$, that can be assigned to each user terminal.
The number of total subchannels assigned to users in each beam is limited by constraint \eqref{eq:constBeamMaxSubs}.
Constraint \eqref{eq:constBeamCoordSelection} implies that each beam can be configured with at most one beam center at one time slot and constraint \eqref{eq:constCoordBEselected} means that each beam center coordinate can be configured to only at most one beam at one time slot to avoid severe inter-beam interference.
Taking into account the on-board power limitations, similar to \cite{choi2005optimumpower}, we impose constraints on the maximum transmission power for each beam and satellite, which are indicated by \eqref{eq:constMaxPowerLmtBeam} and \eqref{eq:constMaxPowerLmt}, respectively.
The requirements on minimum elevation angle and user SINR are ensured by constraint \eqref{eq:constMinEleVAngle} and constraint \eqref{eq:constMinSINR}, respectively.
Constraint \eqref{eq:constBinaryB} and \eqref{eq:constBinaryE} indicate that $\mathbf{b}$ and $\mathbf{e}$ are binary variables.

After elaborating concerned resource management problem, its NP hardness is first proved in the following lemma.

\begin{lemma}
	\label{lemma:NPhard}
	Problem \eqref{eq:primalProblem} is NP-hard.
\end{lemma}

\begin{proof}
	To prove {{Lemma \ref{lemma:NPhard}}}, we consider a special case of the primal problem \eqref{eq:primalProblem}.
	Specifically, we let $|{\mathcal T}| = 1$ and assume there is only one satellite in ${\mathcal M}$.
	Meanwhile, the satellite configures with only one beam $q$, and the beam is with a fixed beam center position $c$.
	As for transmission power constraints \eqref{eq:constMaxPowerLmtBeam} and \eqref{eq:constMaxPowerLmt}, by letting $p_{m,q} = P_q^{\max}$, these two constraints always hold.
	In addition, by setting $\theta _0 = 0$ and $\gamma _0 = 0$, the minimum requirements for elevation angle and user SINR are removed.
	In this special case, the primal problem \eqref{eq:primalProblem} can now be simplified as follows.
	\begin{equation}
		\label{eq:NPproof}
		\begin{aligned}
			\mathop {\max }\limits_{\boldsymbol{\rm{b}}}
			 & \quad \sum_{n\in \mathcal{N}}{U^{\alpha}( \sum_{t\in \mathcal{T}}{\sum_{q\in \mathcal{Q}}{R_{q,n}^{t}}})} \\
			\emph{s.t.} \quad
			 & \sum\limits_{n \in {\mathcal N}} {b_{k,n}} \le 1, \forall k \in {\mathcal K},                             \\
			 & \sum\limits_{n \in {\mathcal N}} {\sum\limits_{k \in {\mathcal K}} {b_{k,n}} } \le K,                     \\
			 & {b_{k,n}} \in \{0,1\}, \forall k \in {\mathcal K}, n \in {\mathcal N}.
		\end{aligned}
	\end{equation}

	It is obvious that the subchannel assignment optimization problem formulated in \eqref{eq:NPproof} is a generalized assignment problem, which has been proven to be NP-hard \cite{ozbakir2010bees}.
	Therefore, the primal problem in \eqref{eq:primalProblem} is also NP-hard.
\end{proof}

\begin{figure}[h]
	\centering
	\includegraphics[width=0.48\textwidth]{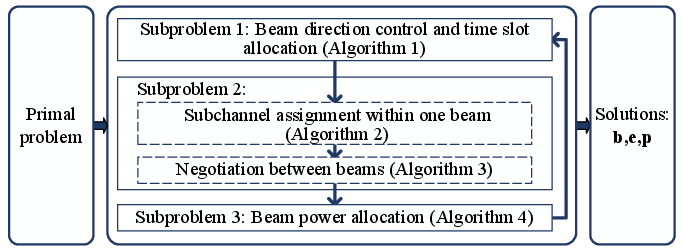}
	\caption{Overview of the proposed resource management framework.}
	\label{fig:des_solutionDesign}
	\vspace{-0.2in}
\end{figure}

\section{Resource Management Framework and Algorithm Design}
\label{sec:solutionDesign}

\subsection{Resource Management Problem Decomposition}

The primal problem \eqref{eq:primalProblem} is a mixed integer non-convex programming problem due to the integer variables $\mathbf{b,e}$ and continuous variable $\mathbf{p}$, as well as the non-convex functions in \eqref{eq:ppOJ}, \eqref{eq:constMinEleVAngle}, and \eqref{eq:constMinSINR}.
In Section \ref{sec:problemForm}, we have proved that problem \eqref{eq:primalProblem} is NP-hard, and it is challenging to find a global optimal solution directly.
To make the problem tractable, problem \eqref{eq:primalProblem} is decoupled into three subproblems, namely beam direction control and time slot allocation, user subchannel assignment, and beam power allocation.
Then, we solve the three subproblems sequentially and iteratively until convergence.
The relation among these subproblems is shown in Fig. \ref{fig:des_solutionDesign}.
In terms of algorithm design, the beam direction control and time slot allocation subproblem as well as user subchannel assignment subproblem will be both solved using matching theory, while beam power allocation subproblem will be handled with successive convex approximation.

\subsection{Problem Formulation and Algorithm Design for Beam Direction Control}

\subsubsection{Problem Formulation}

As stated above, for each beam, its center position is selected from a pre-defined coordinate set and the bore-sight direction of beams can change between time slots \cite{chahat2020cubesatantenna}.
Given user subchannel assignment and beam power allocation schemes, which map variable ${\boldsymbol{\rm{e}}}$ to ${\boldsymbol{\rm{b}}}$ and ${\boldsymbol{\rm{p}}}$, the beam direction control optimization problem across $T$ time slots is expressed as
\begin{equation}
	\label{eq:beamSelectionPart}
	\begin{aligned}
		 & \mathop {\max }\limits_{\boldsymbol{\rm{e}}} \quad
		\sum_{n\in \mathcal{N}}{U^{\alpha}( \sum_{t\in \mathcal{T}}{\sum_{q\in \mathcal{Q}}{R_{q,n}^{t}}})}                                                     \\
		 & \quad \emph{s.t.} \quad  \eqref{eq:constBeamCoordSelection}, \eqref{eq:constCoordBEselected}, \eqref{eq:constMinEleVAngle}, \eqref{eq:constBinaryE}.
	\end{aligned}
\end{equation}

\subsubsection{Reformulation as A Matching Problem}

Note that problem \eqref{eq:beamSelectionPart} is still NP-hard.
To solve this problem with low complexity, we transform it into a two sided many-to-one matching problem \cite{gu2015matchingtheory}, which is defined as follows.
\begin{definition}
	\label{def:BPmatching}
	Define ${\mathcal B}_{BDC} = {\mathcal C} \times {\mathcal T}$ with ${\mathcal C}$ and ${\mathcal T}$ being the set of beam center position candidates and the set of time slots, respectively, and each element $(c,t) \in {\mathcal B}_{BDC}$ represents a unit composed of beam center coordinate $c$ and time slot $t$.
	Then, problem \eqref{eq:beamSelectionPart} can seen as a many-to-one matching between beams whose set is ${\mathcal Q}$ and the elements in set ${\mathcal B}_{BDC}$.
	The matching problem is formally defined as a function $\Phi _{BDC}$, which has the following properties:
	\begin{enumerate}[label=(\arabic*), leftmargin=35pt]
		\item $| {\Phi _{BDC} ((c,t))} | \le 1$ and ${\Phi _{BDC} ((c,t))} \in {\mathcal Q} \cup \emptyset $,
		\item $\forall t^\prime\in \mathcal{T} ,| \{ (c,t)|( c,t ) \in {\Phi _{BDC} (q)},t=t^\prime \} |\le 1$,
		\item ${\Phi _{BDC} ((c,t))} = q$, if and only if $(c,t) \in {\Phi _{BDC} (q)}$.
	\end{enumerate}
\end{definition}
Property (1) corresponds to \eqref{eq:constCoordBEselected} and indicates that each candidate coordinate can be matched with at most one beam during one time slot, property (2) corresponds to \eqref{eq:constBeamCoordSelection} and means that each beam can be configured with at most one center position in one time slot, and property (3) indicates that coordinate-time-slot unit $(c,t)$ matches with beam $q$ if and only if beam $q$ also matches with coordinate-time-slot unit $(c,t)$.

The utility functions of coordinate-time-slot units and beams, which capture the preferences of matching players, should align with the optimization objective in problem \eqref{eq:beamSelectionPart}.
Hence, the utility of unit $(c,t)$ when matching with beam $q$, denoted as $\varphi _{BDC}^{c,t}(\mathcal{N} _{BDC}^{q,c,t})$, is given by
\begin{equation}
	\label{eq:bpUF4coordTime}
	\varphi _{BDC}^{c,t}(\mathcal{N} _{BDC}^{q,c,t})=\sum\limits_{n\in \mathcal{N} _{BDC}^{q,c,t}}{U^{\alpha}({{R}_{q,c,n}^{t}})},
\end{equation}
where $R_{q,c,n}^{t} = \sum\limits_{k \in {{\mathcal K}}} {a_{q,n}^tb_{k,n}^t}R_{q,c,k,n}^t$.
$\mathcal{N} _{BDC}^{q,c,t}$ represents the set of user terminals accessing beam $q$ with the beam center position $c$ at time slot $t$, and it is initialized with the user terminals located within a circular area centered at the beam center position $c$ of beam $q$ with a radius of $r_0$.
This set is updated after optimizing the beam direction control.
For beam $q$, its utility function is designed as follows:
\begin{equation}
	\label{eq:bpUF4beamCoord}
	\varphi _{BDC}^{q}(\mathcal{B} _{BDC}^{q})= \sum\limits_{n\in \mathcal{N}} U^{\alpha}({\sum\limits_{(c,t) \in \mathcal{B} _{BDC}^{q}}{R_{q,c,n}^{t}}}),
\end{equation}
where $\mathcal{B} _{BDC}^{q}= \Phi _{BDC}\left( q \right) $ is the set of coordinate-time-slot units matching with beam $q$.

\begin{algorithm}[t]
	\caption{Beam Direction Control Matching Algorithm}
	\label{alg:MABP}
	\begin{algorithmic}[1]
		\REQUIRE{\bf{1: \emph{Initialization}}}
		\STATE Construct preference lists of the beams $\varOmega _{BDC}^{q}, \forall q \in {\mathcal{Q}}$, and the units $\varOmega _{BDC}^{c,t}, \forall (c,t) \in {\mathcal{B}_{BDC}}$;
		\STATE Construct the set of the coordinate-time-slot units that are not matched ${\mathcal{B}_{BDC}^0}$;
		\STATE Set the index of iteration $r=0$, the units sets of accepted by beams $\mathcal{B}_{BDC}^{q,0} = \emptyset, \forall q \in {\mathcal{Q}}$;
		\WHILE{${\mathcal{B}_{BDC}^0} \ne \emptyset$ and $\exists (c,t) \in {\mathcal{B}_{BDC}^0}, \varOmega _{BDC}^{c,t} \ne \emptyset$}
		\STATE $r= r+1$;
		\FOR{$\forall (c,t) \in {\mathcal{B}_{BDC}^0}$}
		\STATE Find $q=\arg \max _{q\in \varOmega _{BDC}^{c,t}}\varphi _{BDC}^{c,t}( \mathcal{N} _{BDC}^{q,c,t}) $, and propose to beam $q$;
		\ENDFOR
		\FOR{$\forall q \in {\mathcal{Q}}$}
		\STATE Denote the units who propose to beam $q$ as $\mathcal{B}_{BDC}^{q'}$, and form $\mathcal{S} =\mathcal{B} _{BDC}^{q'}\cup \mathcal{B} _{BDC}^{q,r-1}$;
		\STATE Beam $q$ keeps most preferred unit at each time slot in $\mathcal{S}$ to form $\mathcal{S}'$ according to $\mathcal{S}'=\{ ( c,t ) =\arg \max _ {( c,t ) \in \{\mathcal{S} ,t=t' \}}\varphi _{BDC}^{q}( (c,t) ) ,\forall t'\in \mathcal{T} \} $, and then accepts the first $\min \{T,|\mathcal{S}'|\}$ best ranked units in $\mathcal{S}'$ to update $\mathcal{B}_{BDC}^{q,r}$;
		\STATE Remove the matched units from ${\mathcal{B}_{BDC}^0}$, and add the rejected units to ${\mathcal{B}_{BDC}^0}$;
		\STATE Remove $q$ from the preference lists of units that have sent proposals;
		\ENDFOR
		\ENDWHILE
		\REQUIRE{\bf{2: \emph{Swap Matching}}}
		\STATE Set $s_{i,j} = 0, i,j \in \mathcal{B}_{BDC}$;
		\STATE For any unit $i \in \mathcal{B}_{BDC}$, it searches for another unit $j$ to check the existence of swap-blocking pair;
		\WHILE{there exists swap-blocking pair}
		\IF{$\left< i,j \right> $ is a swap-blocking pair and ${\scalemath{0.9}{s_{i,j} + s_{j,i} < I_1}}$}
		\STATE $\Phi _{BDC} = \Phi _{BDC}^{i,j}$ and $s_{i,j}=s_{i,j}+1$;
		\ENDIF
		\ENDWHILE
	\end{algorithmic}
\end{algorithm}

\subsubsection{Matching-based Beam Direction Control and Time Slot Allocation Algorithm}

Based on utility function \eqref{eq:bpUF4coordTime}, each coordinate-time-slot unit $(c,t)$ is able to form its own preference list $\varOmega _{BDC}^{c,t}$ with regard to beams.
Similarly, each beam $q$ is able to construct its preference list $\varOmega _{BDC}^{q}$ with regard to coordinate-time-slot units based on utility function \eqref{eq:bpUF4beamCoord}.
To indicate their preferences, a preference relation $\succ $ is introduced for both sides of matching players.
Specifically, $q \succ _{(c,t)}q'$ means that unit $(c,t)$ prefers beam $q$ over beam $q'$ if and only if $\varphi _{BDC}^{c,t}(\mathcal{N}_{BDC}^{q,c,t})>\varphi _{BDC}^{c,t}(\mathcal{N}_{BDC}^{q',c,t})$.
Similarly, $\mathcal{B} _{BDC}^{q}\succ _q {\mathcal{B}^\prime} _{BDC}^{q}$ means beam $q$ prefers the set of units ${{\mathcal{B}} _{BDC}^{q}}$ over ${\mathcal{B}^\prime} _{BDC}^{q}$ if and only if $\varphi _{BDC}^{q}(\mathcal{B} _{BDC}^{q})>\varphi _{BDC}^{q}({\mathcal{B}^\prime} _{BDC}^{q})$.

Adhering to the classical deferred-acceptance matching procedure \cite{roth1999twosidedmatching}, each coordinate-time-slot unit proposes to its top-preferred beam that has not previously rejected it, guided by individual preferences.
The beams then accept proposals from the most preferred coordinate-time-slot units and reject the others.
However, due to the existence of inter-beam co-channel interference, the utility of each beam $q$ in \eqref{eq:bpUF4beamCoord} is not only related to the individual matching state of itself but also is influenced by the matching state of other beams, and this phenomenon is called \emph{externalities} in matching theory \cite{bodine-baron2011peereffects}.
This kind of matching problem cannot be directly tackled by the traditional deferred acceptance algorithm as discussed in \cite{bodine-baron2011peereffects}.

To deal with the interdependencies between the preferences of matching players, \emph{swap matching} operation is involved and is formally defined as
\begin{equation}
	\label{eq:BPswapMatching}
	\begin{aligned}
		\Phi _{BDC}^{i,j} =
		 & \{ \Phi _{BDC} \backslash \{ (i,\Phi _{BDC} (i)),(j,\Phi _{BDC} (j))\} \} \\
		 & \cup \{ (i,\Phi _{BDC} (j)),(j,\Phi _{BDC} (i))\},
	\end{aligned}
\end{equation}
where two coordinate-time-slot units $i,j \in {\mathcal{B}}_{BDC}$ exchange their matched pairs and all other matching pairs remain unchanged.
According to the above \emph{swap matching} definition, we further introduce \emph{swap-blocking pair} concept as follows.
\begin{definition}
	\label{def:BPswapBlockPair}
	A pair of coordinate-time-slot units $\left< i,j \right> $ is a swap-blocking pair if:
	\begin{enumerate}[label=(\arabic*), leftmargin=20pt]
		\item $i$ and $j$ represent $(c_1,t_1 )$ and $( c_2,t_2)$, respectively, with $t_1 = t_2$, and $q_1=\Phi _{BDC}\left( i \right) $, $q_2=\Phi _{BDC}\left( j \right) $,
		\item $\forall s\in \{i,j,q_1,q_2\},\varphi _s(\Phi _{BDC}^{i,j})\ge \varphi _s(\Phi _{BDC})$,
		\item $\exists s\in \{i,j,q_1,q_2\}$, such that ${\varphi _s}(\Phi _{BDC}^{i,j}) > {\varphi _s}({\Phi _{BDC}})$,
		\item $\sum_{q\in \mathcal{Q} \backslash \left\{ q_1,q_2 \right\}}{\varphi _q\left( \Phi _{BDC}^{i,j} \right)}\ge \sum_{q\in \mathcal{Q} \backslash \left\{ q_1,q_2 \right\}}{\varphi _q\left( \Phi _{BDC} \right)}$,
	\end{enumerate}
	where ${\varphi _s}(\Phi _{BDC})$ denotes the utility of player $s$ under matching state $\Phi _{BDC}$.
\end{definition}
{\bf{Definition \ref{def:BPswapBlockPair}}} indicates that if two coordinate-time-slot units want to switch between two beams, the beams involved must approve the swap, and vice versa.
Property (1) indicates that the swap operation is only possibly triggered when the two units are located within the same time slot.
Property (2) states that the swap operation can not decrease the utility of any involved player.
Property (3) implies that the swap operation must increase the utility of at least one player.
Property (4) suggests that the total utility of the other beams can not decrease to avoid fluctuations in the objective.

To address the matching problem for beam direction control and time slot allocation, we propose a swap-operation based matching algorithm that includes two phases, as shown in {{Algorithm \ref{alg:MABP}}}.
In Phase 1, a deferred acceptance based procedure is implemented with utility functions that do not consider inter-beam interference to generate an initial matching state, while
Phase 2 involves swap operations to further update the matching state with utility functions incorporating inter-beam interference.

\subsubsection{Property Analysis of Algorithm \ref{alg:MABP}}
\label{sec:propertyAnalysisBP}

For the proposed {{Algorithm \ref{alg:MABP}}}, its properties in terms of convergence, stability, and complexity are analyzed.
First, algorithm convergence is proved in the following proposition.
\begin{proposition}
	\label{prop:BPconvergece}
	{Algorithm \ref{alg:MABP}} converges to a final matching $\Phi _{BDC}^*$ after a limited number of swap operations.
\end{proposition}
\begin{proof}
	\label{proof:BPconvergece}
	According to {\bf{Definition \ref{def:BPswapBlockPair}}}, for a given swap-blocking pair $\left< i,j \right>$ under matching $\Phi _{BDC}$ with $q = \Phi _{BDC}(i)$ and $q' = \Phi _{BDC}(j)$, the swap matching operation must satisfy $\varphi_{q}(\Phi _{BDC}^{i,j}) \ge \varphi_{q}(\Phi _{BDC})$, $\varphi_{q'}(\Phi _{BDC}^{i,j}) \ge \varphi_{q'}(\Phi _{BDC})$, and $\sum_{q\in \mathcal{Q} \backslash \left\{ q_1,q_2 \right\}}{\varphi _q\left( \Phi _{BDC}^{i,j} \right)}\ge \sum_{q\in \mathcal{Q} \backslash \left\{ q_1,q_2 \right\}}{\varphi _q\left( \Phi _{BDC} \right)}$.
	Therefore, when the swap matching operation is executed in Phase 2 of {{Algorithm \ref{alg:MABP}}}, the total utility of beams will strictly increase.
	On the other hand, the number of feasible swap-blocking pairs is finite due to the limited number of matched beams, and there exists an upper bound for the total utility of beams due to the limited power and spectrum resources.
	As a result, {{Algorithm \ref{alg:MABP}}} will converge to a final matching after a limited number of swap operations.
\end{proof}

As for the stability of the final matching, due to the existence of externalities, the traditional ``pairwise-stability'' \cite{roth1999twosidedmatching} concept is no longer suitable \cite{ni2021resourceallocation}.
Therefore, we focus on the \emph{two-sided exchange-stable} concept for swap matching proposed in \cite{bodine-baron2011peereffects}, which is defined as follows.
\begin{definition}
	\label{def:BPstability}
	A many-to-one matching $\Phi _{BDC}$ is {{two-sided exchange-stable}} if and only if there is no swap-blocking pair \cite{bodine-baron2011peereffects}.
\end{definition}
Based on the above definition, the stability of {{Algorithm \ref{alg:MABP}}} is proved in the following proposition.

\begin{proposition}
	\label{prop:BPfinalStable}
	The final matching ${\Phi }^*_{BDC}$ obtained by {{Algorithm \ref{alg:MABP}}} is two-sided exchange-stable.
\end{proposition}
\begin{proof}
	\label{proof:BPfinalStable}
	As ${\Phi }^*_{BDC}$ represents a final matching, it follows, in accordance with {Proposition \ref{prop:BPconvergece}}, that the total utility of beams cannot be improved.
	Assuming the existence of a swap-blocking pair $\left<i,j \right>$, it follows that, according to {Definition \ref{def:BPswapBlockPair}} and Phase 2 of {Algorithm \ref{alg:MABP}}, the total utility of beams could be further improved via a swap operation.
	However, this contradicts the convergence conditions of the final matching stipulated in {Proposition \ref{prop:BPconvergece}}.
	Consequently, no swap-blocking pair is present in the final matching, rendering it two-sided exchange stable in accordance with Definition 3.
\end{proof}

\begin{proposition}
	\label{prop:BPcomplexity}
	The complexity of {{Algorithm \ref{alg:MABP}}} is upper bounded by $\mathcal{O}(CTQ^2 + I_1(CT)^2)$.
\end{proposition}

\begin{proof}
	\label{proof:BPcomplexity}
	As shown in {{Algorithm \ref{alg:MABP}}}, two phases are involved.
	The complexity of Phase 1 is dominated by setting up preference lists for two sided players, which is $\mathcal{O} \left( CTQ^2 \right) $.
	For Phase 2, it is hard to derive the closed form of the number of swap operations executed when the algorithm converges.
	Since parameter $s_{i,j}$ is defined to void excessive swap operations between players, the number of iterations in Phase 2 is up to $I_1 \times \left( CT \right) ^2$ in the worst case.
	Hence, the complexity of Phase 2 is upper bounded by $\mathcal{O}(I_1(CT)^2)$.
	Therefore, the total complexity of {{Algorithm \ref{alg:MABP}}} is $\mathcal{O}(CTQ^2 + I_1(CT)^2)$.
\end{proof}

\subsection{Problem Formulation and Algorithm Design for User Subchannel Assignment}
\subsubsection{Problem Formulation}
For the given beam direction $\boldsymbol{\rm{e}}$ and the power allocation scheme, the optimal subchannel assignment can be determined by solving the following optimization problem
\begin{equation}
	\label{eq:subchAssPart}
	\begin{aligned}
		 & \mathop {\max }\limits_{\boldsymbol{\rm{b}}} \quad
		\sum_{n\in \mathcal{N}}{U^{\alpha}( \sum_{t\in \mathcal{T}}{\sum_{q\in \mathcal{Q}}{R_{q,n}^{t}}})}                        \\
		 & \emph{s.t.} \quad  \eqref{eq:constSubcAssi}-\eqref{eq:constBeamMaxSubs},\eqref{eq:constMinSINR},\eqref{eq:constBinaryB}
	\end{aligned}
\end{equation}

\subsubsection{Reformulation as A Matching Problem}
In the following, we first design a subchannel assignment algorithm for user terminals within one beam and then extend it to multi-beam scenarios.
Similar to beam direction control and time slot allocation, we model subchannel assignment subproblem as a two sided many-to-one matching problem.

\begin{definition}
	\label{def:SAmatching}
	With the combination of the subchannels ${\mathcal K}$ and time slots ${\mathcal T}$, the subchannel-time-slot set is constructed as ${\mathcal B}_{SA} = {\mathcal K} \times {\mathcal T}$, where each element $(k,t)$ denotes a subchannel-time-slot unit.
	Then, subchannel assignment can be transformed into a many-to-one matching problem between user terminals in ${\mathcal N}$ and
	subchannel-time-slot units in ${\mathcal B}_{SA}$, which has the following properties
	\begin{enumerate}[label=(\arabic*), leftmargin=35pt]
		\item $| {\Phi _{SA} ((k,t))} | \le 1$ and ${\Phi _{SA} ((k,t))} \in {\mathcal N} \cup \emptyset $,
		\item $\forall t^{\prime}\in \mathcal{T} ,|\{(k,t)|(k,t)\in \Phi _{SA}(n),t=t^{\prime}\}|\le K^{thr}$,
		\item ${\Phi _{SA} ((k,t))} = n$, if and only if $(k,t) \in {\Phi _{SA} (n)}$.
	\end{enumerate}
\end{definition}

Property (1) corresponds to constraint \eqref{eq:constSubcAssi} and means that each subchannel within one beam can only be allocated to one user terminal.
Property (2) corresponds to constraint \eqref{eq:constSubcAssiMax} and indicates that the number of subchannels within one beam assigned to a user terminal at each time slot must not exceed $K^{thr}$.
Property (3) means that if user terminal $n$ matches with subchannel-time-slot unit $(k,t)$, and then $(k,t)$ also matches with user terminal $n$.

We define the utility functions of user terminals and subchannel-time-slot units as follows.
For user terminal $n$, its utility function is given by
\begin{equation}
	\label{eq:saUF4UE}
	\varphi _{SA}^{n}(\mathcal{B} _{SA}^{n}) = {U^{\alpha}}(\sum\limits_{(k,t)\in \mathcal{B} _{SA}^{n}}R_{q,k,n}^{t}),
\end{equation}
where $R_{q,k,n}^{t} = \sum\limits_{c \in {{\mathcal C}}} {a_{q,n}^t}{e_{c,q}^t}R_{q,c,k,n}^t$ with ${a_{q,n}^t}$ and ${e_{c,q}^t}$ being determined by the beam direction control scheme, and $\mathcal{B} _{SA}^{n}$ is the set of subchannel-time-slot units assigned to user terminal $n$.
For subchannel-time-slot unit $(k,t)$, its utility function is defined as follows
\begin{equation}
	\label{eq:saUF4Beam}
	\varphi _{SA}^{k,t}(n) = U^{\alpha}({R_{q,k,n}^t}),
\end{equation}
where ${\mathcal{N}_{SA}^{q,t}}$ is the set of user terminals associated with beam $q$ at time slot $t$.

\begin{algorithm}[t]
	\caption{Subchannel Assignment Matching Algorithm within One Beam}
	\label{alg:MA4SAoneBeam}
	\begin{algorithmic}[1]
		\STATE Construct preference lists of the units $\scalemath{0.92}{\varOmega _{SA}^{k,t}, \forall (k,t) \in {\mathcal{B}_{SA}}}$, and the users ${\varOmega _{SA}^{n}, \forall n \in {\mathcal{N}}}$;
		\STATE Construct the set of the subchannel-time-slot units that are not matched ${\mathcal{B}_{SA}^0}$;
		\STATE Set the index of iteration $r=0$, the units sets of accepted by user $n$, $\mathcal{B}_{SA}^{n,0} = \emptyset, \forall n \in {\mathcal{N}}$;
		\WHILE{${\mathcal{B}_{SA}^0} \ne \emptyset$ and $\exists (k,t) \in {\mathcal{B}_{SA}^0}, \varOmega _{SA}^{k,t} \ne \emptyset$}
		\STATE $r= r+1$;
		\FOR{$\forall (k,t)\in {\mathcal{B}_{SA}^0}$}
		\STATE Find $n=\arg \max _{n\in \varOmega _{SA}^{k,t}}\varphi _{SA}^{k,t}( n )$, and send a proposal to user $n$;
		\ENDFOR
		\FOR{$\forall n \in {\mathcal{N}}$}
		\STATE Denote the units who propose to user $n$ as $\mathcal{B}_{SA}^{n'}$, and form $\mathcal{S} =\mathcal{B} _{SA}^{n'}\cup \mathcal{B} _{SA}^{n,r-1}$;
		\STATE User $n$ keeps the first $K^{thr}$ preferred unit at each time slot in $\mathcal{S}$ to form $\mathcal{S}'$;
		\STATE User $n$ accepts the first $\min \{K^{thr}T,|\mathcal{S}'|\}$ best ranked units to update $\mathcal{B}_{SA}^{n,r}$;
		\STATE Remove the matched units from ${\mathcal{B}_{SA}^0}$, and add the rejected units to ${\mathcal{B}_{SA}^0}$;
		\STATE Remove $n$ from the preference lists of units that have sent proposals;
		\ENDFOR
		\ENDWHILE
	\end{algorithmic}
\end{algorithm}

\begin{algorithm}[t]
	\caption{Subchannel Assignment Matching Algorithm for All Beams}
	\label{alg:MA4SAbeams}
	\begin{algorithmic}[1]
		\REQUIRE{\bf{1: \emph{Initialization}}}
		\FOR{$\forall q \in \mathcal{Q}$}
		\STATE Obtain the set of subchannels assigned to user $n$ by beam $q$ at time slot $t$, $\mathcal{K}_{SA}^{n,q,t}, n \in \mathcal{N}, t \in {\mathcal{T}}$ by {{Algorithm \ref{alg:MA4SAoneBeam}}};
		\ENDFOR
		\REQUIRE{\bf{2: \emph{Negotiation}}}
		\FOR{$\forall t \in {\mathcal{T}}$}
		\STATE Set $s_{q,k}=0, \forall q \in {\mathcal{Q}}, k \in {\mathcal{K}}$;
		\WHILE{there exists interfering beam pair $\left< q,q' \right>$, and $\exists k \in \mathcal{K}_{\rm{inter}}^{q,q'}, s_{q,k} + s_{q',k} < I_2$}
		\STATE Obtain the set of interference subchannels $\mathcal{K}_{\rm{inter}}^{q,q'}$ according to {\bf{Definition} \ref{def:SAbeamPair}};
		\FOR{$\forall k \in \mathcal{K}_{\rm{inter}}^{q,q'}$ and $s_{q,k} + s_{q',k} < I_2$}
		\STATE Find $n_1=\arg \max_{n\in \mathcal{N} _{SA}^{q,t}} \varphi _{SA}^{k,t}( n,q )$, and then $n_2=\arg \max_{n\in \mathcal{N} _{SA}^{q',t}} \varphi _{SA}^{k,t}( n,q' )$;
		\STATE Find $\left( \bar{n},\bar{q} \right) =\arg\min _{i\in \left\{ \left( n_1,q \right) ,\left( n_2,q\prime \right) \right\}}\varphi _{SA}^{k,t}(i)$;
		\STATE Remove $k$ from $\mathcal{K} _{SA}^{\bar{n},\bar{q},t}$;
		\STATE Set $s_{\bar{q},k}=s_{\bar{q},k}+1$;
		\ENDFOR
		\ENDWHILE
		\ENDFOR
	\end{algorithmic}
\end{algorithm}

\subsubsection{Matching-based Subchannel Assignment Algorithm}
In this part, we present the deferred acceptance based matching algorithm for subchannel assignment within a single beam, as shown in {Algorithm \ref{alg:MA4SAoneBeam}}.
Notably, this algorithm does not consider inter-beam interference.
Upon obtaining the matching state between subchannel-time-slot units and user terminals within each beam using {Algorithm \ref{alg:MA4SAoneBeam}}, a negotiation procedure is further needed for beam pairs targeting the reused subchannels to address inter-beam interference.
As inter-beam interference arises only when a subchannel is used by multiple beams within the same time slot, we consider performing the negotiation for the matching state between subchannel-time-slot units and user terminals at each time slot individually.
For instance, consider a pair of beams interfering with each other in subchannel $k$ at time slot $t$, denoted as $\left<q_1,q_2\right>$.
If the removal of subchannel $k$ from the allocated subchannel set of one beam results in a total utility of subchannel-time-slot units improvement, both beams are inclined to accept this negotiation proposal.
The corresponding operation is called \emph{negotiation operation}, which is defined as
\begin{equation}
	\label{eq:SAremoval}
	\begin{split}
		\Phi _{SA}^{q_1,q_2,k,t}=
		& \left\{ \Phi _{SA}^{q_1,t},\Phi _{SA}^{q_2,t}\backslash \left\{ \left( k,\Phi _{SA}^{q_2,t}\left( k \right) \right) \right\} \right\} \\
		&\cup \left\{ \Phi _{SA}^{q,t},\forall q\in \mathcal{Q} \backslash \left\{ q_1,q_2 \right\} \right\},
	\end{split}
\end{equation}
where $\Phi _{SA}^{q_1,q_2,k,t}$ denotes the matching state changed from matching $\Phi _{SA}^{t}$ by removing subchannel $k$ from the set of allocated subchannels of beam $q_2$ at time slot $t$, and $\Phi _{SA}^{t} = \left\{ \Phi _{SA}^{q,t},\forall q\in \mathcal{Q} \right\} $ with $\Phi _{SA}^{q,t}$ being the matching state between subchannel-time-slot units and user terminals within beam $q$ at time slot $t$.
Moreover, we define $s_{q,k}$ to count the negotiation times for subchannel $k$ of beam $q$ at time slot $t$ to avoid excessive negotiation operations.
To identify which beams need to perform the \emph{negotiation operation}, we define \emph{interfering beam pair} as follows.
\begin{definition}
	\label{def:SAbeamPair}
	A pair of beams $\left< q_1,q_2 \right>$ is an interfering beam pair at time slot $t$ if they satisfy
	\begin{enumerate}[label=(\arabic*), leftmargin=35pt]
		\item ${\exists {n} \in {\mathcal{N}_{SA}^{q_1,t}}, \theta _{q_2,{n},ele}^{t} \ge \theta_0}$, and
		\item ${\exists k\in \mathcal{K} _{SA}^{{n},q_1,t}, k\notin \mathcal{K} _{\rm{unused}}^{q_2,t}, h_{q_2,{n},k}^{t}p_{q_2,k}^{t}\ge I_{\mathrm{inter}}^{0}}$, and
		\item ${\scalemath{0.9}{ \sum\limits_{q\in \mathcal{Q}}{\sum\limits_{k\in \mathcal{K} ^{q,t}}{\varphi _{SA}^{k,t}\left( \Phi _{SA}^{q_1,q_2,k,t} \right)}}>\sum\limits_{q\in \mathcal{Q}}{\sum\limits_{k\in \mathcal{K} ^{q,t}}{\varphi _{SA}^{k,t}\left( \Phi _{SA}^{t} \right)}}}}$, or

		      ${\scalemath{0.9}{\sum\limits_{q\in \mathcal{Q}}{\sum\limits_{k\in \mathcal{K} ^{q,t}}{\varphi _{SA}^{k,t}\left( \Phi _{SA}^{q_2,q_1,k,t} \right)}}>\sum\limits_{q\in \mathcal{Q}}{\sum\limits_{k\in \mathcal{K} ^{q,t}}{\varphi _{SA}^{k,t}\left( \Phi _{SA}^{t} \right)}}}}$,
	\end{enumerate}
	where $\mathcal{K}_{SA}^{n,q,t}$ is the set of subchannels assigned to user $n$ by beam $q$ at time slot $t$, ${\mathcal{K}_{\rm{unused}}^{q,t}}$ denotes the set of unused subchannels to any user by beam $q$ at time slot $t$, and ${\mathcal{K}^{q,t}} = {\mathcal{K}} \backslash {\mathcal{K}_{\rm{unused}}^{q,t}}$ represents the subchannels allocated to user terminals within beam $q$.
\end{definition}
Based on {Definition \ref{def:SAbeamPair}}, the negotiation process is designed in {{Algorithm \ref{alg:MA4SAbeams}}}.

\subsubsection{Property Analysis of Algorithm \ref{alg:MA4SAbeams}}
Based on Definition \ref{def:SAbeamPair}, the convergence and stability proofs for Algorithm \ref{alg:MA4SAbeams} can be referred to as the corresponding proofs of Proposition \ref{prop:BPconvergece} and Proposition \ref{prop:BPfinalStable} for Algorithm \ref{alg:MABP}, which are omitted here for brevity.
\begin{proposition}
	\label{prop:complexitySAMA}
	The complexity of {{Algorithm \ref{alg:MA4SAbeams}}} is upper bounded by $\mathcal{O}(QKTN^2+{TKQ(Q-1)}I_2)$.
\end{proposition}
\begin{proof}
	\label{proof:complexitySAMA}
	In {{Algorithm \ref{alg:MA4SAbeams}}}, Phase 1 involves $Q$ independent subchannel assignment problems for all beams, which are handled using the deferred acceptance based {Algorithm \ref{alg:MA4SAoneBeam}}.
	The complexity for a single beam is $\mathcal{O}(KTN^2)$.
	Consequently, the complexity of Phase 1 in {{Algorithm \ref{alg:MA4SAbeams}}} is $\mathcal{O}(QKTN^2)$.
	Regarding Phase 2, the number of iterations required for convergence is difficult to determine.
	However, an upper bound on the complexity of Phase 2 can be derived, as each subchannel of one beam is limited to participating in the negotiation process at most $I_2$ times, which implies that the complexity of the negotiation process among all beams at a single time slot is $\mathcal{O}({KQ(Q-1)}I_2)$.
	The complexity of Phase 2 across $T$ time slots in {{Algorithm \ref{alg:MA4SAbeams}}} is $\mathcal{O}({TKQ(Q-1)I_2})$.
	Thus, the upper bound of {{Algorithm \ref{alg:MA4SAbeams}}} complexity is $\mathcal{O}(QKTN^2+{TKQ(Q-1)}I_2)$.
\end{proof}

\subsection{Problem Formulation and Algorithm Design for Beam Power Allocation}

Given beam direction control and user subchannel assignment schemes, the beam power allocation optimization problem simplified from \eqref{eq:primalProblem} is formulated as
\begin{equation}
	\label{eq:powerAllocationPart}
	\begin{aligned}
		 & \mathop {\max }\limits_{\boldsymbol{\rm{p}}} \quad
		\sum_{n\in \mathcal{N}}{U^{\alpha}( \sum_{t\in \mathcal{T}}{\sum_{q\in \mathcal{Q}}{R_{q,n}^{t}}})}                \\
		 & \quad \emph{s.t.} \quad  \eqref{eq:constMaxPowerLmtBeam}, \eqref{eq:constMaxPowerLmt},\eqref{eq:constMinPower}.
	\end{aligned}
\end{equation}
Note that problem \eqref{eq:powerAllocationPart} is non-convex due to inter-beam interference in \eqref{eq:subchannelSINR}.
To tackle this nonconvexity, a successive convex approximation approach is employed to transform problem \eqref{eq:powerAllocationPart} into a convex one that can be solved effectively.

To apply the SCA approach, we first rearrange the objective function in \eqref{eq:powerAllocationPart} by following Jenssen's inequality \cite{boyd2004convexoptimizatio}
\begin{equation}
	\label{eq:PAobRearrange}
	\sum_{n\in \mathcal{N}}{U^{\alpha}( \sum_{t\in \mathcal{T}}{\sum_{q\in \mathcal{Q}}{R_{q,n}^{t}}})}\ge \sum_{n\in \mathcal{N}}{\sum_{t\in \mathcal{T}}{\sum_{q\in \mathcal{Q}}{U^{\alpha}\left( R_{q,n}^{t} \right)}}}.
\end{equation}
To handle the highly nonconcave rate function \eqref{eq:subchannelDataRate}, we resort to the widely used logarithmic approximation \cite{papandriopoulos2009scalelowcomplexity} to derive the following lower bound:
\begin{equation}
	\label{eq:PAlowBound}
	\log _2\left( 1+\gamma _{q,c,k,n}^{t} \right) \ge \mu _{q,c,k,n}^{t}\log _2\left( \gamma _{q,c,k,n}^{t} \right) +\upsilon _{q,c,k,n}^{t},
\end{equation}
with
\begin{subequations}
	\label{eq:logAppro}
	\begin{equation}
		\label{eq:mudefine}
		\mu  _{q,c,k,n}^{t}=\frac{\tilde{\gamma}_{q,c,k,n}^{t}}{1+\tilde{\gamma}_{q,c,k,n}^{t}}
	\end{equation}
	and
	\begin{equation}
		\begin{aligned}
			\label{eq:upsilondefine}
			\upsilon _{q,c,k,n}^{t} = & \log _2\left( 1+\tilde{\gamma}_{q,c,k,n}^{t} \right)                    \\
			                          & - \mu _{q,c,k,n}^{t}\log _2\left( \tilde{\gamma}_{q,c,k,n}^{t} \right).
		\end{aligned}
	\end{equation}
\end{subequations}
When $\tilde{\gamma}_{q,c,k,n}^{t} = \gamma _{q,c,k,n}^{t}$, the equivalence of \eqref{eq:PAlowBound} is achieved.

Based on the above lower bound approximation and let $\hat{\boldsymbol{\rm{p}}}=\ln (\boldsymbol{\rm{p}})$, we obtain the approximation of $R_{q,c,k,n}^{t}$ as
\begin{equation}
	\label{eq:PAlowBoundApprox}
	\begin{aligned}
		 & R_{q,c,k,n}^{t} \ge \tilde{R}_{q,c,k,n}^{t}\left( e^{\hat{\boldsymbol{\rm{p}}}} \right)                                                                           \\
		 & =  \sum_{k\in \mathcal{K}}{ B_{q,k}( {\mu _{q,c,k,n}^{t}}\frac{\ln( \gamma _{q,c,k,n}^{t}(e^{\hat{\boldsymbol{\rm{p}}}}))}{{\ln(2)}} +\upsilon _{q,c,k,n}^{t} )}.
	\end{aligned}
\end{equation}
As a result, we obtain the following approximated problem
\begin{equation}
	\label{eq:PAapprox}
	\begin{aligned}
		 &
		\mathop {\max }\limits_{{\boldsymbol{\rm{\hat{p}}}}} \quad
		\sum_{n\in \mathcal{N}}{\sum_{t\in \mathcal{T}}{\sum_{q\in \mathcal{Q}}{U^{\alpha}\left( \tilde{R}_{q,n}^{t}\left( e^{\hat{\boldsymbol{\rm{p}}}} \right) \right)}}}
		\\
		 & \quad \emph{s.t.} \quad  \eqref{eq:constMaxPowerLmtBeam}, \eqref{eq:constMaxPowerLmt},\eqref{eq:constMinPower}.
	\end{aligned}
\end{equation}

\begin{algorithm}[t]
	\caption{SCA based Power Allocation Iterative Algorithm}
	\label{alg:SCA4PA}
	\begin{algorithmic}[1]
		\REQUIRE{\bf{1: \emph{Initialization}}}
		\STATE $r=0$, $\boldsymbol{\rm{p}}=\boldsymbol{\rm{p}}^{(0)}$;
		\STATE Compute the objective value of problem \eqref{eq:powerAllocationPart} as $F^{(0)}$;
		\REQUIRE{\bf{2: \emph{Update}}}
		\WHILE{$F^{(r)}-F^{(r-1)}\le \varepsilon _{\rm{thr}}$}
		\STATE $r = r+1$;
		\STATE $\tilde{\gamma}_{q,c,k,n}^{t} ={\gamma}_{q,c,k,n}^{t,{(r-1)}}$;
		\STATE Update $\boldsymbol{\rm{\mu}}^{(r)}$ and $\boldsymbol{\rm{\upsilon}}^{(r)}$ by \eqref{eq:logAppro};
		\STATE Solve the convex optimization problem \eqref{eq:PAapprox} to get the power allocation solution $\boldsymbol{\rm{\hat{p}}}^{(r)}$;
		\STATE Obtain $\boldsymbol{\rm{{p}}}^{(r)}$ by $\boldsymbol{\rm{{p}}}^{(r)}=e^{\boldsymbol{\rm{\hat{p}}}^{(r)}}$;
		\STATE Compute ${\gamma}_{q,c,k,n}^{t,{(r)}}$ and ${F^{(r)}}$;
		\ENDWHILE
	\end{algorithmic}
\end{algorithm}

\begin{proposition}
	\label{prop:PAconvex}
	The approximated problem \eqref{eq:PAapprox} is a concave maximization problem.
\end{proposition}
\begin{proof}
	\label{proof:PAconvex}
	In the right side of \eqref{eq:PAlowBoundApprox}, $\ln( \gamma _{q,c,k,n}^{t}(e^{\hat{\boldsymbol{\rm{p}}}}))$ can be rearranged as
	\begin{equation}
		\label{eq:logRearrange}
		\begin{aligned}
			 & \ln ( \gamma _{q,c,k,n}^{t}(e^{\hat{\boldsymbol{\rm{p}}}}))
			=\ln\left( h_{q,c,k,n}^{t} \right) + {\hat{p}_{q,k}^{t}}                                                                                                                  \\
			 & - \ln( {\sum_{q'\in \mathcal{Q} \backslash q} \sum_{n'\in \mathcal{N}} {{a_{q',n'}^tb_{k,n'}^te_{c',q'}^t}h_{q',c',k,n}^{t}e^{\hat{p}_{q',k}^{t}}}}+\sigma _{n}^{2} ).
		\end{aligned}
	\end{equation}
	Because the log-sum-exp function is convex \cite{boyd2004convexoptimizatio}, we can conclude that \eqref{eq:logRearrange} is a concave function and then the approximation rate function $\tilde{R}_{q,n}^{t}( e^{\hat{\boldsymbol{\rm{p}}}})$ is a concave function.
	Since $\alpha$-proportional utility function is an increasing strictly concave function for any given $\alpha$, $U^{\alpha}\left( \tilde{R}_{q,n}^{t}\left( e^{\hat{\boldsymbol{\rm{p}}}} \right) \right) $ is a concave function according to \cite{boyd2004convexoptimizatio}.
	Obviously, the objective function in \eqref{eq:PAapprox} is the sum of concave terms, which is also a concave function.
	Consequently, problem \eqref{eq:PAapprox} is a concave maximization problem.
\end{proof}

Problem \eqref{eq:PAapprox} can be directly tackled by interior point method, which is readily available in most optimization toolboxes, such as CVX.
Note that the optimal solution obtained in \eqref{eq:PAapprox} is a lower bound of the objective function of \eqref{eq:powerAllocationPart}.
By solving problem \eqref{eq:PAapprox} to obtain ${\gamma}_{q,c,k,n}^{t}$, we iteratively update $\mu  _{q,c,k,n}^{t}$ and $\upsilon _{q,c,k,n}^{t}$ according to \eqref{eq:mudefine} and \eqref{eq:upsilondefine}, respectively, and then tighten the bound in \eqref{eq:PAlowBoundApprox} to eventually solve problem \eqref{eq:powerAllocationPart}.
At the beginning of the first iteration, i.e., $r=0$, we initiate power allocation variable $\boldsymbol{\rm{p}}^{(0)}$.
In the $r$-th iteration, we set $\tilde{\gamma}_{q,c,k,n}^{t} = {\gamma}_{q,c,k,n}^{t} (\boldsymbol{\rm{p}}^{(r-1)})$, where $\boldsymbol{\rm{p}}^{(r-1)}$ is the optimal solution of problem \eqref{eq:PAapprox} at the previous iteration.
The details of the proposed SCA based power allocation algorithm are shown in {{Algorithm \ref{alg:SCA4PA}}}.

\begin{proposition}
	\label{prop:PAIAconvergence}
	The SCA based Algorithm \ref{alg:SCA4PA} will finally converge to a locally optimal solution to problem \eqref{eq:powerAllocationPart}.
\end{proposition}
\begin{proof}
	\label{proof:PAIAconvergence}
	In the $r-$th iteration of {{Algorithm \ref{alg:SCA4PA}}}, the following relationship can be derived:
	\begin{equation}
		\label{eq:PAiteration}
		\begin{aligned}
			 & F( \tilde{\boldsymbol{R}}( \boldsymbol{\mathrm{p}}^{( r )},\boldsymbol{\mathrm{\mu}} ^{( r )},
			\boldsymbol{\mathrm{\upsilon}} ^{( r )} ) )  =\underset{\boldsymbol{\mathrm{p}}}{\max} \,\,F( \tilde{\boldsymbol{R}}( \boldsymbol{\mathrm{p}},\boldsymbol{\mathrm{\mu}} ^{( r )},\boldsymbol{\mathrm{\upsilon}} ^{( r )} ) ) \\
			 & \overset{(a)}{\ge} F( \tilde{\boldsymbol{R}}( \boldsymbol{\mathrm{p}}^{( r-1 )},\boldsymbol{\mathrm{\mu}} ^{( r )},\boldsymbol{\mathrm{\upsilon}} ^{( r )} ) )                                                            \\
			 & \overset{(b)}{=} F( \boldsymbol{R}( \boldsymbol{\mathrm{p}}^{( r-1 )} ) )                                                                                                                                                 \\
			 & \overset{(c)}{\ge} F( \tilde{\boldsymbol{R}}( \boldsymbol{\mathrm{p}}^{( r-1 )},\boldsymbol{\mathrm{\mu}} ^{( r-1 )},\boldsymbol{\mathrm{\upsilon}} ^{( r-1 )} ) ),
		\end{aligned}
	\end{equation}
	where $F({\boldsymbol{R}})$ and $F(\tilde{\boldsymbol{R}})$ denote the objective value of problem \eqref{eq:powerAllocationPart} and problem \eqref{eq:PAapprox}, respectively.
	Specifically, $(a)$ holds due to the fact that $\boldsymbol{\mathrm{p}}^{( r )}$ is the optimal solution to problem \eqref{eq:PAapprox} under $\boldsymbol{\mathrm{\mu}} ^{( r )}$ and $\boldsymbol{\mathrm{\upsilon}} ^{( r )}$, $(b)$ holds for the reason that $\boldsymbol{\mathrm{\mu}} ^{( r )}$ and $\boldsymbol{\mathrm{\upsilon}} ^{( r )}$ are calculated from the optimal solution $\boldsymbol{\mathrm{p}}^{( r -1 )}$, and $(c)$ holds according to the definition of the lower bound approximation in \eqref{eq:PAlowBoundApprox}.
	Therefore, the approximated value $F( \tilde{\boldsymbol{R}}( \boldsymbol{\mathrm{p}}^{( r -1)},\boldsymbol{\mathrm{\mu}} ^{( r-1 )},
		\boldsymbol{\mathrm{\upsilon}} ^{( r-1 )} ) )$ increases after each iteration by utilizing the optimal solution $\boldsymbol{\mathrm{p}}^{( r-1)}$ of problem \eqref{eq:PAapprox}.
	Moreover, the feasible region of problem \eqref{eq:powerAllocationPart} is compact and the corresponding objective value is upper bounded due to the total power and spectrum resource constraints, which implies the SCA based {{Algorithm \ref{alg:SCA4PA}}} will finally converge to a solution $\boldsymbol{\mathrm{p}}^*$.
	According to \cite{papandriopoulos2009scalelowcomplexity}, solution $\boldsymbol{\mathrm{p}}^*$ is a local optimal solution that satisfies the necessary KKT optimality conditions of problem \eqref{eq:powerAllocationPart}.
\end{proof}

\begin{proposition}
	\label{prop:complexitySCA}
	The complexity of {{Algorithm \ref{alg:SCA4PA}}} is upper bounded by $\mathcal{O} (\overline{r}\left( TMQ \right) ^{3.5})$.
\end{proposition}
\begin{proof}
	The computational complexity of {{Algorithm \ref{alg:SCA4PA}}} is dominated by solving problem \eqref{eq:PAapprox}.
	By employing the interior point method, the computational complexity of problem \eqref{eq:PAapprox} is $\mathcal{O} (\left( TMQ \right) ^{3.5})$.
	Consequently, the overall computational complexity of {{Algorithm \ref{alg:SCA4PA}}} is expressed as $\mathcal{O} (\overline{r}\left( TMQ \right) ^{3.5})$, with $\overline{r}$ representing the number of iterations required to fulfill the convergence condition.
\end{proof}

\subsection{Summary of Our Proposed Resource Management Approach}
\label{sec:summaryAlgorithm}

The proposed resource management approach consists of four iterative algorithms.
In the first iteration, the beam power allocation and subchannel assignment strategies are initiated, followed by the execution of Algorithm \ref{alg:MABP} to obtain feasible beam center positions for each beam, as shown in Fig. \ref{fig:des_solutionDesign}.
Based on the obtained beam center positions and the initiated beam power allocation strategy, Algorithm \ref{alg:MA4SAoneBeam} is employed to obtain a feasible subchannel assignment strategy for all time slots without considering inter-beam interference. Subsequently, Algorithm \ref{alg:MA4SAbeams} is applied to coordinate subchannel allocation among beams at each time slot to mitigate inter-beam interference.
At the end of this iteration, Algorithm \ref{alg:SCA4PA} updates the beam power allocation strategy based on the obtained beam center positions and subchannel assignment strategy.
Then, the following iteration will use the beam power allocation and subchannel assignment strategies obtained from the previous iteration to start the iteration and execute {Algorithm \ref{alg:MABP}}.
The iteration of these four algorithms continues until convergence is reached.

\begin{table}[]
	\caption{Main Simulation Parameters}
	\centering
	\label{tab:simulationSettings}
	\begin{tabular}{c|c}
		\toprule[1pt]
		\textbf{Parameter}                             & \textbf{Value}           \\ \midrule[1pt]
		The altitude of satellite orbits               & $780\,km$                \\
		Orbit inclination                              & $45^\circ$               \\
		The number of orbits                           & 16                       \\
		The number of satellites per orbit             & 30                       \\
		Configuration period                           & $100\,s$                 \\
		Length of time slot $\tau$                     & $1\,s$                   \\
		The diameter of target ground area             & $500\,km$                \\
		The number of beam center candidates $C$       & 200                      \\
		The number of satellites for the area $M$      & 2                        \\
		The number of users in the area $N$            & 50                       \\
		Radius for beam service users initiation $r_0$ & $100\,km$                \\
		User receiving antenna gain $G_n^{{\rm{rx}}}$  & $39.7\,\rm{dBi}$         \\
		The diameter of satellite antenna $D$          & $0.5\,m$                 \\
		Satellite antenna aperture efficiency $\eta$   & 0.65                     \\
		$\rho_1, \rho_2, \rho_3$                       & 0.95, 0.1, 0.058         \\
		Beam bandwidth $B$                             & $400\,\rm{MHz}$          \\
		Carrier Frequency $f$                          & $20\,\rm{GHz}$ (Ka band) \\
		Noise temperature                              & $150\,\rm{K}$            \\
		Maximal beam transmission power $P_q^{max}$    & $200\,\rm{W}$            \\
		Maximal satellite power $P_m^{max}$            & $1200\,\rm{W}$           \\
		Minimum elevation angle $\theta_0$             & $25^\circ$               \\
		$I_1, I_2$                                     & 2,2                      \\
		\bottomrule[1pt]
	\end{tabular}
\end{table}

\section{Simulation Results}
\label{sec:numericalRes}
In this section, we first describe the simulation environment and present the setting of simulation parameters.
Then, simulation results are provided to demonstrate the performance of the multi-beam satellite network with the proposed joint beam direction control and resource allocation scheme.
To simulate real satellite communication constellations, a Walker constellation of LEO satellites is constructed by AGI Systems Tool Kit (STK) and used in all simulations.
The parameters of the constellation and other main simulation parameters are listed in Table \ref{tab:simulationSettings}.

In terms of the satellite constellation employed in our experiments, it consists of 16 orbits at an altitude of 780 $km$, with each orbit characterized by an inclination of 45$^\circ$ and accommodating 30 satellites.
The targeted service area, centrally located at coordinates (41.7642$^\circ$N, 86.6513$^\circ$E), spans a radius of 500 $km$ and contains $N = 50$ users.
The potential beam center positions are distributed within the area, with a total of $C=200$ candidate coordinates.
The configuration period is set as $T = 100$ seconds from 14 Oct 2022 04:02:00.000 UTCG to 14 Oct 2022 04:03:40.000 UTCG, during which $M = 2$ satellites of the constellation are able to cover the target service area simultaneously, with a time slot length of $\tau = 1$ second.
In each time slot, satellite positions are updated based on their movement trajectory, and the network topology at the 50-th time slot is illustrated in Fig. \ref{fig:des_UEdistributionRandom}.
The maximal beam transmission power and satellite power are configured at $P_q^{max} = 200 W$ and $P_m^{max} = 1200 W$, respectively.
The minimum required elevation angle is set at 25$^\circ$, consistent with the setting of Starlink \cite{pachler2021updatedcomparison}.
For downlink communication between satellites and users, the Ka band is utilized, and each beam has a maximum bandwidth of $400$ MHz.

To show the effectiveness of our proposed resource management approach for dynamic multi-beam multi-satellite downlink networks, the following two schemes are considered as baselines:
\begin{enumerate}
	\item {\bf{Baseline 1}}:
	      In this baseline scheme, all beams are configured with fixed directions that are user cluster centers, which are determined by the clustering algorithm that divides user terminals into $Q$ clusters.
	      The user subchannel assignment is handled by matching theory, and the beam power allocation is handled by SCA approach, which are the same as our proposal.
	\item {\bf{Baseline 2}}: In this baseline scheme, the transmission power of each satellite is equally distributed among all of its beams at each time slot, and both beam direction and subchannel assignment are determined by matching based algorithms, which are the same as our proposal.
\end{enumerate}

\begin{figure}[t]
	\centering
	\includegraphics[width=0.48\textwidth]{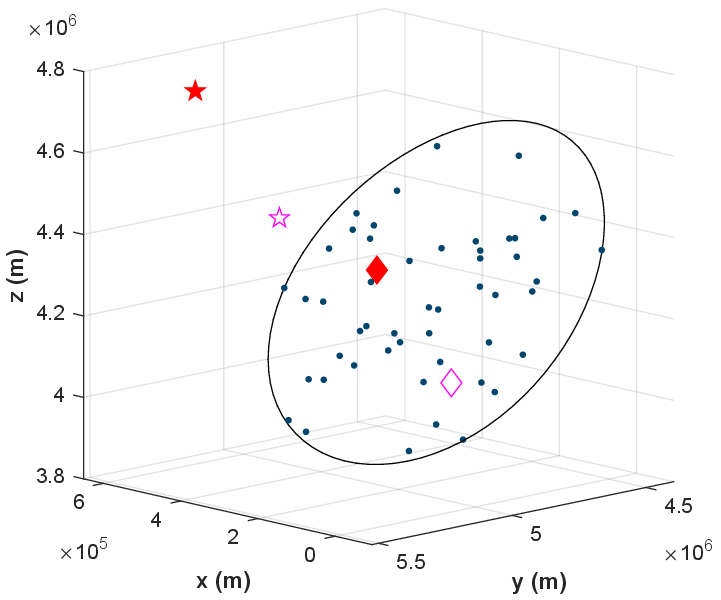}
	\caption{The network topology at the 50-$th$ time slot. The points represent user terminals, the pentagrams represent satellites, and the diamonds represent sub-satellite points.}
	\label{fig:des_UEdistributionRandom}
\end{figure}

\subsection{The Verification of Convergence}

In Fig. \ref{fig:GPJBPRAconvergence}, the evaluation of the objective in \eqref{eq:primalProblem}, i.e., the sum of user $\alpha$-utility, is demonstrated under different number of beams $L$ at each satellite, where the number of subchannels per beam $K=20$ and the maximum number of subchannels that can be allocated to a single user $K^{thr}=6$.
We can observe that our proposed resource management approach, as summarized in Section \ref{sec:summaryAlgorithm}, converges within a small number of outer iterations of Algorithm \ref{alg:MABP}.
Specifically, it converges after the first iteration when $L=1$.
When $L = 7$, it converges in less than 5 iterations.
This is because increasing $L$ means that there are more beams whose center positions and transmission power should be determined in each time slot, which leads to more iterations to converge.

\begin{figure}[t]
	\centering
	\includegraphics[width=0.48\textwidth]{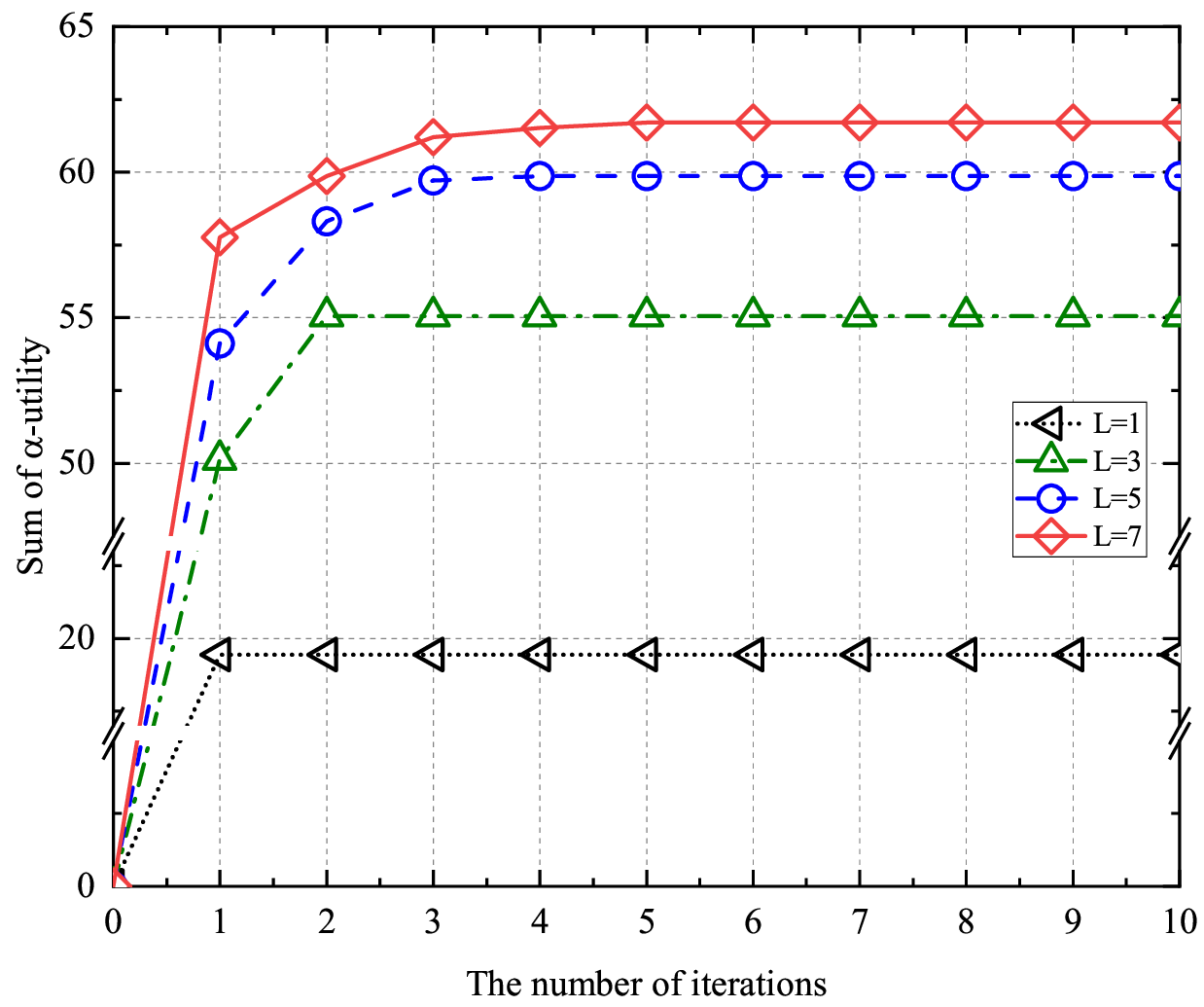}
	\caption{The convergence of our proposal with different numbers of beams configured at each satellite ($K=20, K^{thr}=6$).}
	\label{fig:GPJBPRAconvergence}
\end{figure}

\begin{figure*}[t]
	\centering
	\subfigure[]{
		\label{fig:GPnumBeams1}
		\includegraphics[width=0.31\linewidth]{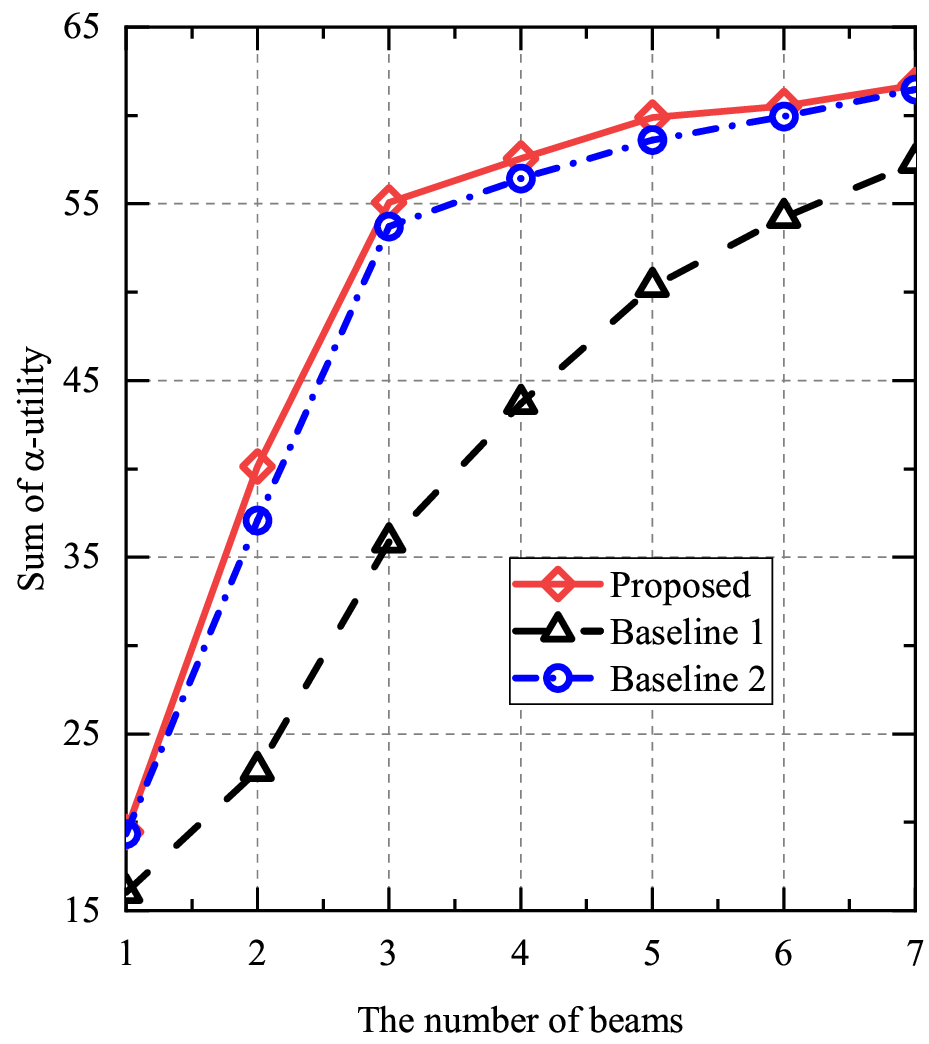}}
	\subfigure[]{
		\label{fig:GPnumBeams2}
		\includegraphics[width=0.31\linewidth]{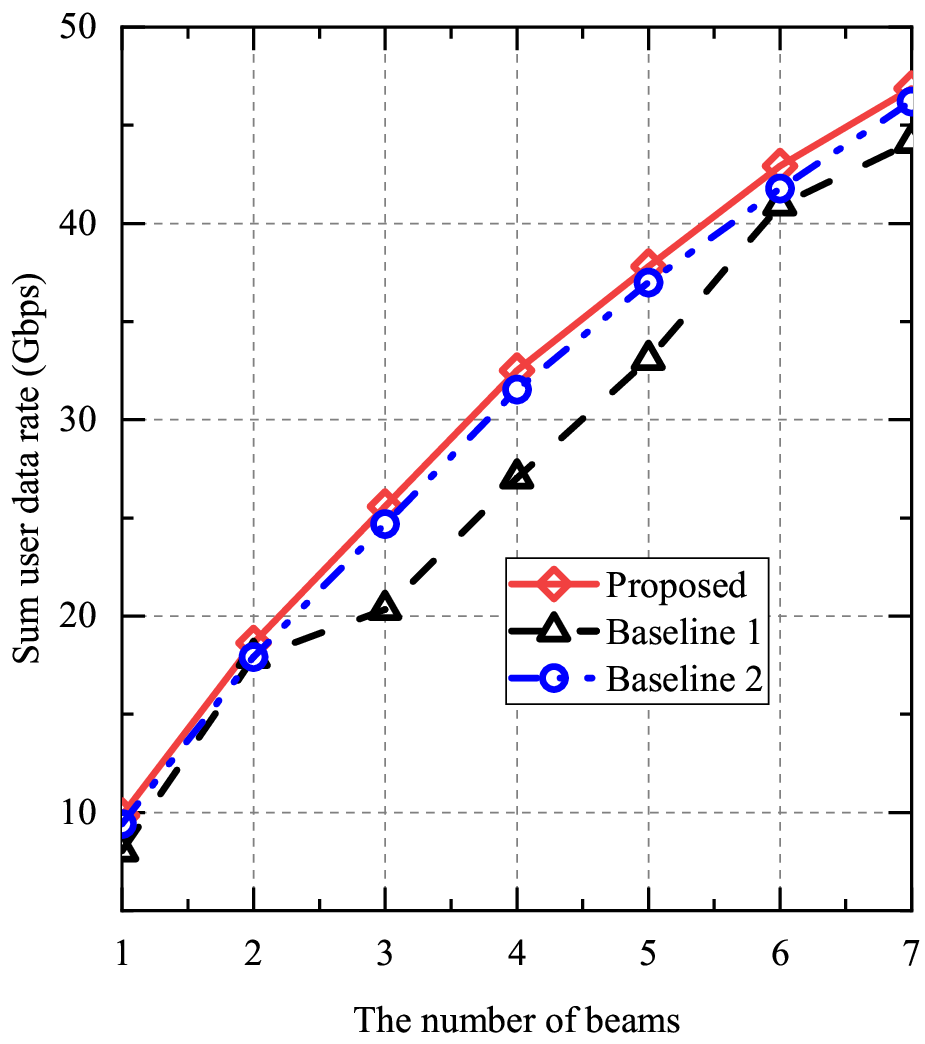}}
	\subfigure[]{
		\label{fig:GPnumBeams3}
		\includegraphics[width=0.31\linewidth]{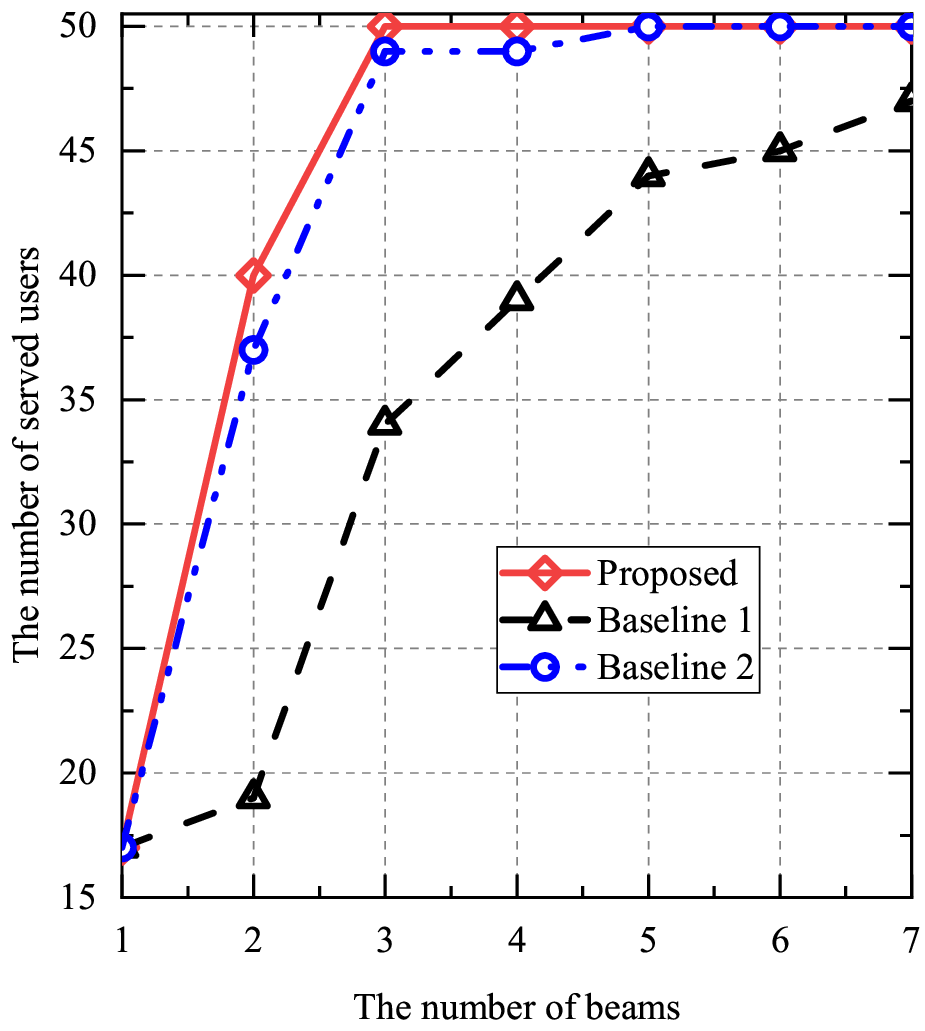}}
	\caption{Impact of the number of beams configured at each satellite on the network performance ($K=20, K^{thr}=6$). (a): The number of beams \emph{vs.} sum of $\alpha$-utility; (b) The number of beams \emph{vs.} sum user data rate; (c) The number of beams \emph{vs.} the number of served users.}
	\label{fig:GPnumBeams}
	\vspace{-0.1in}
\end{figure*}

\begin{figure}[t]
	\centering
	\includegraphics[width=0.48\textwidth]{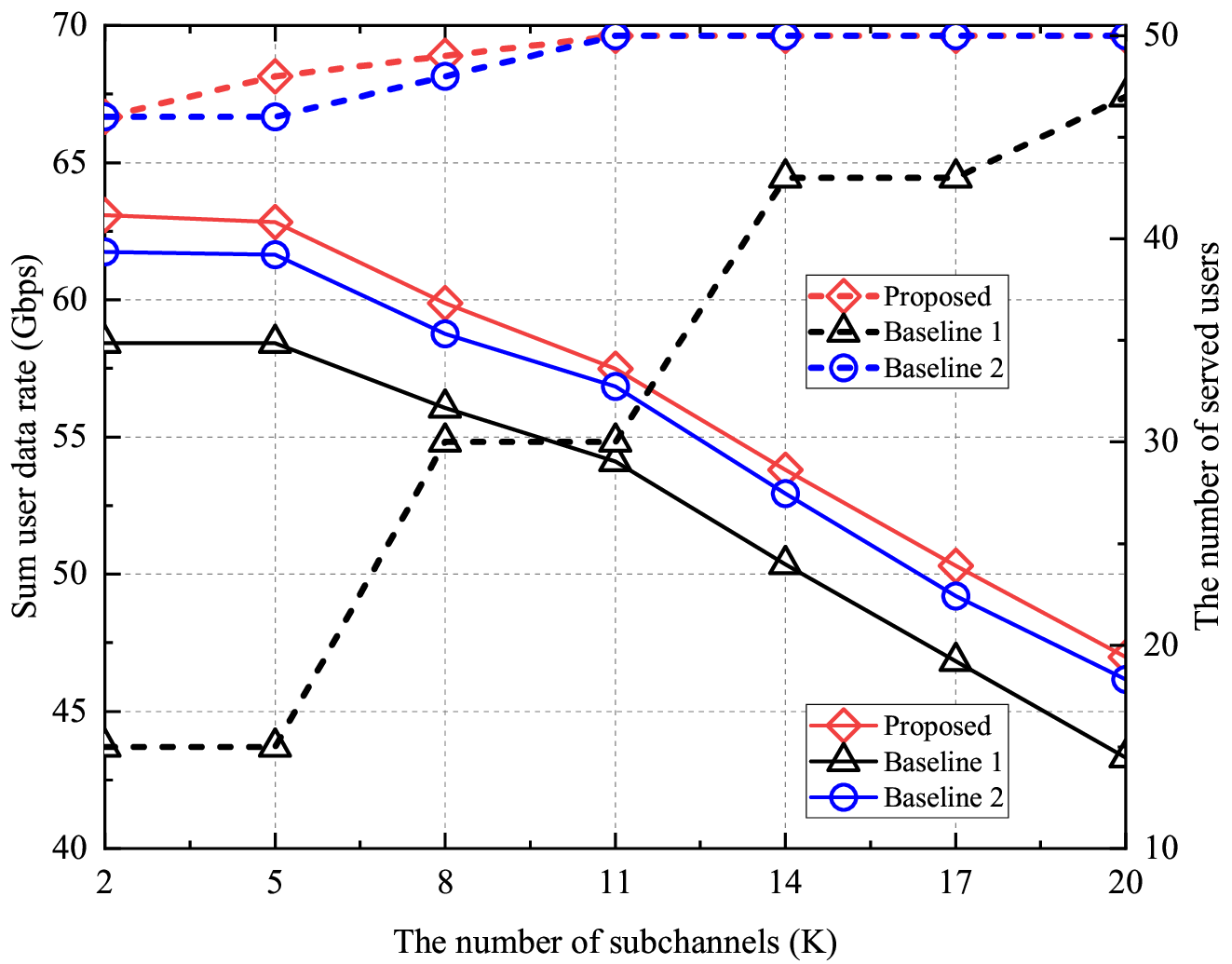}
	\caption{Impact of the number of subchannels per beam on the network performance ($L=7,K^{thr}=6$). The solid lines correspond to the sum user data rate while the dash lines correspond to the number of severed users.}
	\label{fig:GPnumberSub}
\end{figure}
\begin{figure}[t]
	\centering
	\includegraphics[width=0.48\textwidth]{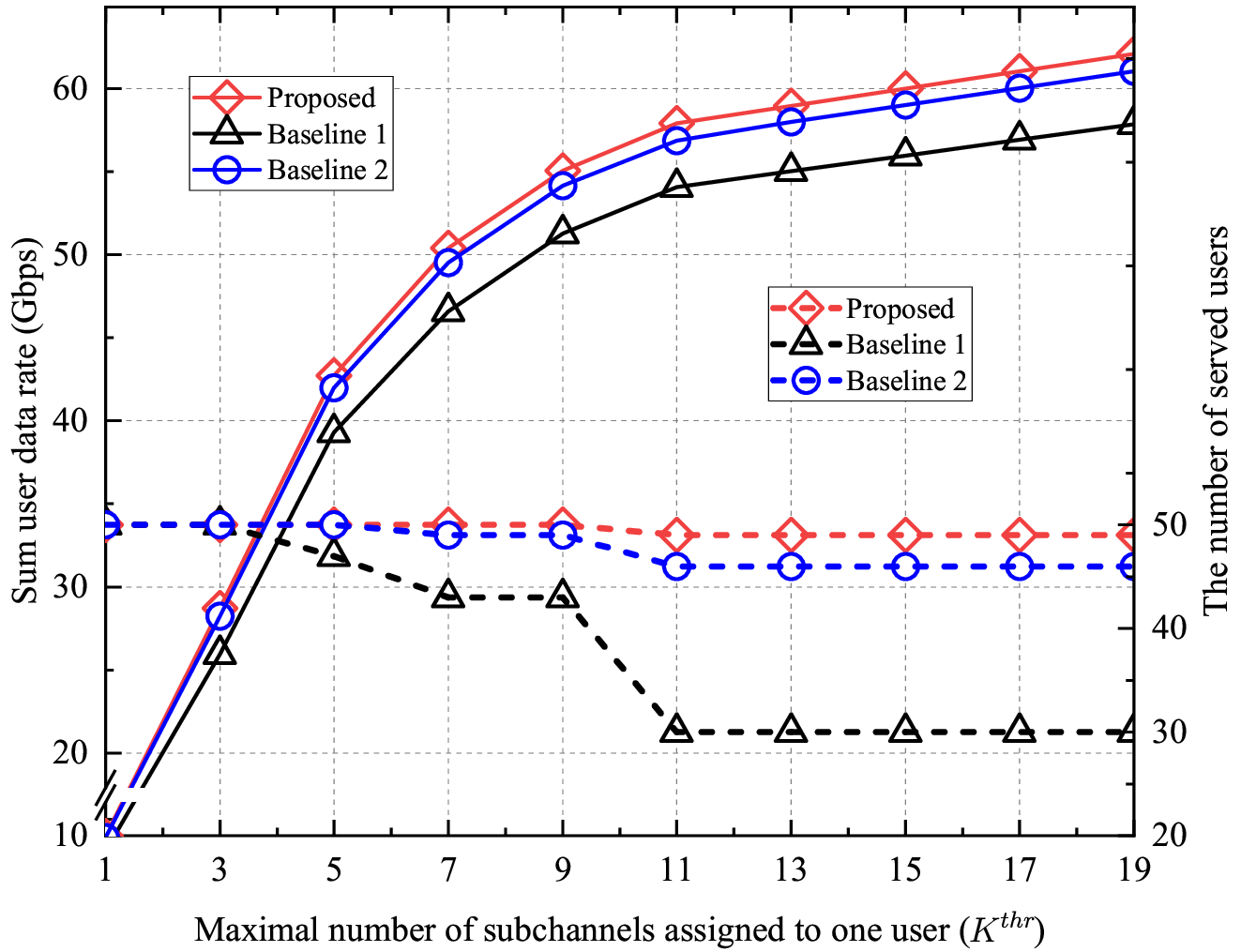}
	\caption{Impact of the maximum number of subchannels that can be allocated to a single user on the network performance ($L=7, K=20$). The solid lines correspond to the sum user data rate while the dash lines correspond to the number of severed users.}
	\label{fig:GPnumberMaxk}
\end{figure}

\subsection{The Impacts of the Number of Beams Configured at Each Satellite}

In Fig. \ref{fig:GPnumBeams}, our proposed resource management approach is compared to two baseline schemes in terms of sum $\alpha$-utility in \eqref{eq:primalProblem}, user sum data rate $\sum_{n\in \mathcal{N}}{\sum_{t\in \mathcal{T}}{\sum_{q\in \mathcal{Q}}{\frac{R_{q,n}^{t}}{T}}}}$, and the number of severed users during the configuration period under varying numbers of beams configured at each satellite.
Here, the number of subchannels per beam is $K=20$, and the maximum number of subchannels that can be allocated to a single user is $K^{thr}=6$.
It is evident that a higher number of beams at each satellite contributes to an increase in the sum of user $\alpha$-utility.
This is because the total power of the satellite used for communication, as well as the probability of a user accessing a beam with a better channel state, both increase when the number of beams configured at each satellite increases.
Moreover, our proposal outperforms the two baselines in terms of all the metrics.
This superior performance compared with Baseline 1 can be attributed to sum $\alpha$-utility being directly considered in the beam direction control of our proposal, whereas Baseline 1 determines beam centers based only on the geographic distribution of users.
In comparison to Baseline 2, our proposal with power allocation optimization can better alleviate inter-beam interference, thereby enhancing network performance.
Furthermore, due to the improved user fairness brought about by $\alpha$-utility in beam direction control, our proposal and Baseline 2 can serve nearly all users when only three beams are equipped at each satellite, i.e., $L=3$, as shown in Fig. \ref{fig:GPnumBeams3}.
Additionally, when $L$ exceeds 3, the increase in the sum of $\alpha$-utility tends to decelerate in Fig. \ref{fig:GPnumBeams1}, but the sum user data rate continues to rise gradually in Fig. \ref{fig:GPnumBeams2}.
This is because there is no significant increase in the number of served users as $L$ increases, as indicated in Fig. \ref{fig:GPnumBeams3}.
Upon further observation, it is evident that optimizing beam direction can yield a more substantial increase in the sum user data rate, up to 25\%, in comparison to the 5\% increase achieved through beam power allocation optimization.
This finding offers pivotal insights for the design of resource allocation schemes in multi-beam satellite communication systems.

\subsection{The Impacts of the Number of Subchannels per Beam}

Fig. \ref{fig:GPnumberSub} illustrates the sum user data rate and the number of served users versus the number of subchannels per beam, where the number of beams configured at each satellite is $L=7$ and the maximum number of subchannels that can be allocated to a single user is $K^{thr}=6$.
First, it is observed that the sum user data rate decreases monotonically when the number of subchannels per beam increases.
This is due to the fact that the increase in the number of subchannels per beam will decrease the available bandwidth and transmission power of each subchannel and the data rate of a single user can be decreased under the fixed maximal number of subchannels.
Second, it can be seen that the sum user data rate remains nearly unchanged when $K<K^{thr}$.
The reason is that the subchannels tend to be assigned to the users with better channel states.
Third, it is shown that the number of served users increases monotonically as the number of subchannels per beam increases.
This is attributed to the higher probability of assigning subchannels to users with less favorable channel states as the parameter $K$ increases.
Still, with the user fairness taken into account and SCA based beam power optimization, our proposal outperforms that of both baseline schemes.
Particularly, compared to Baseline 1, our proposal results in a two-fold increase in the number of served users.

\begin{figure}[t]
	\centering
	\includegraphics[width=0.48\textwidth]{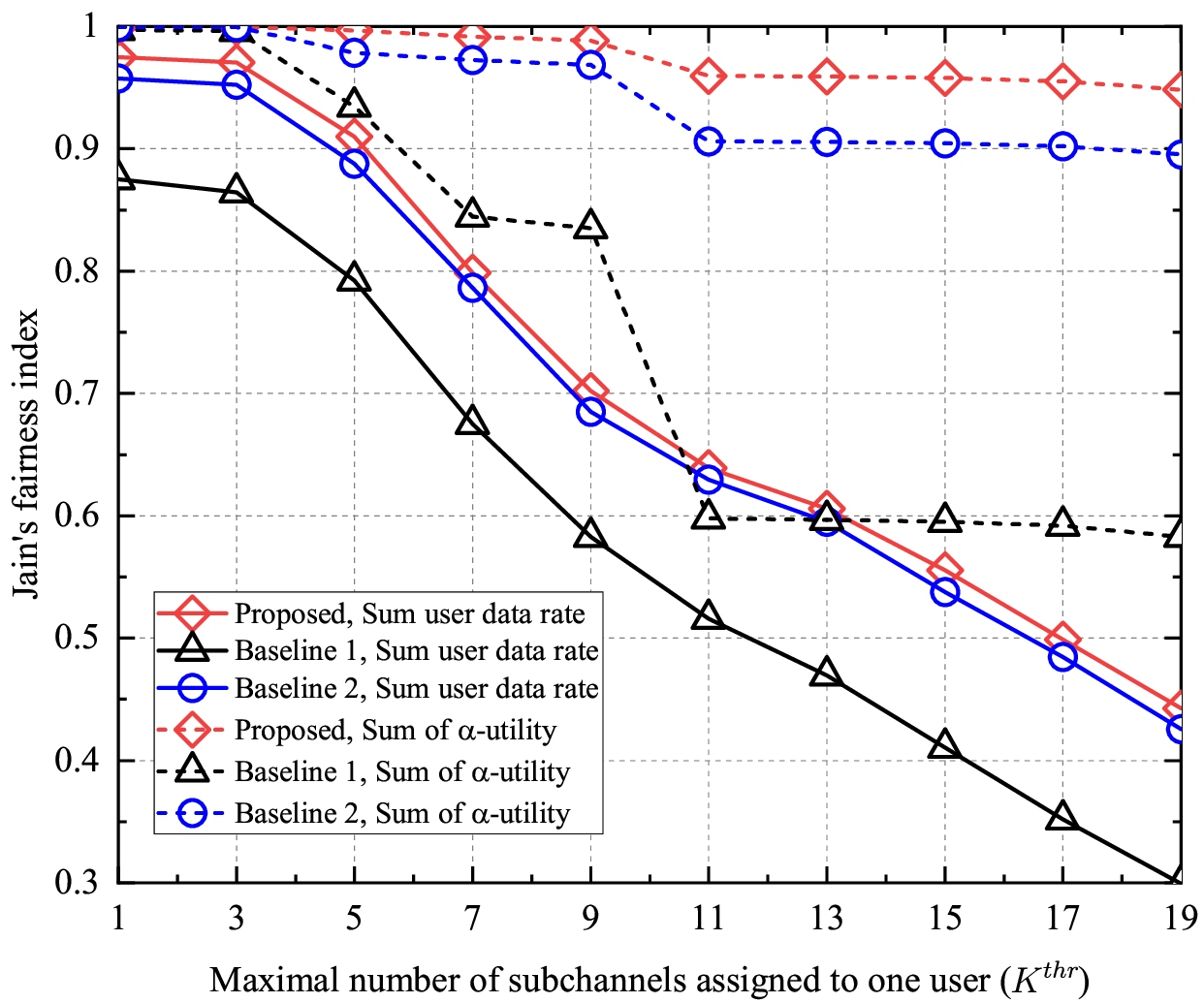}
	\vspace{-0.1in}
	\caption{Impact of maximum number of subchannels per user allocated on the user fairness ($L=7, K=20$).}
	\label{fig:GPJFI}
	\vspace{-0.1in}
\end{figure}

\subsection{The Impacts of the Maximum Number of Subchannels per User}

In Fig. \ref{fig:GPnumberMaxk}, we investigate the impact of the maximum number of subchannels that can be allocated to a single user $K^{thr}$ on the sum user data rate and the number of served users, where the number of beams configured at each satellite is $L=7$ and the number of subchannels per beam is $K=20$.
First, one can observe that the sum user data rate increases monotonically as $K^{thr}$ increases when the number of subchannels per beam is fixed.
This trend is similar to that shown in Fig. \ref{fig:GPnumberSub}, which is due to the fact that assigning most subchannels to users with the better channel states leads to higher network performance.
Correspondingly, as $K^{thr}$ increases, the number of served users decreases.

\subsection{The Evaluation of User Fairness}

Fig. \ref{fig:GPJFI} further evaluates the fairness among all users with the maximum number of subchannels each user can occupy varying, where the number of beams configured at each satellite is $L=7$ and the number of subchannels per beam is $K=20$.
Here, Jain's fairness index (JFI) \cite{jain1984quantitativemeasure} is utilized, which is expressed as $\text{JFI} = \frac{\left( \sum_{n\in \mathcal{N}}{a_n} \right) ^2}{N\times \sum_{n\in \mathcal{N}}({{a_n}^2})}$, where $a_n = \sum_{t\in \mathcal{T}}{\sum_{q\in \mathcal{Q}}{R_{q,n}^{t}}}$ for the sum user data rate and $a_n = U^{\alpha}(\sum_{t\in \mathcal{T}}{\sum_{q\in \mathcal{Q}}{R_{q,n}^{t}}})$ for the sum user $\alpha$-utility.
The value range of JFI is between 0 and 1.
When JFI gets closer to 1, a higher degree of fairness among users is reached.
Fig. \ref{fig:GPJFI} demonstrates that the fairness index of our proposal decreases as the maximum number of subchannels each user can occupy increases.
This is because users with better channel states are more likely to be allocated more subchannels in order to achieve higher network performance.
As a result, fairness among users will significantly decrease when $K^{thr}\le K$.
Meanwhile, it is also worth noting that our proposed resource management scheme is able to significantly improve JFI compared to the other two baselines.

\begin{figure}[t]
	\centering
	\subfigcapskip=-5pt
	\subfigure[]{
		\label{fig:des_UEdistribution50km}
		\includegraphics[width=0.8\linewidth]{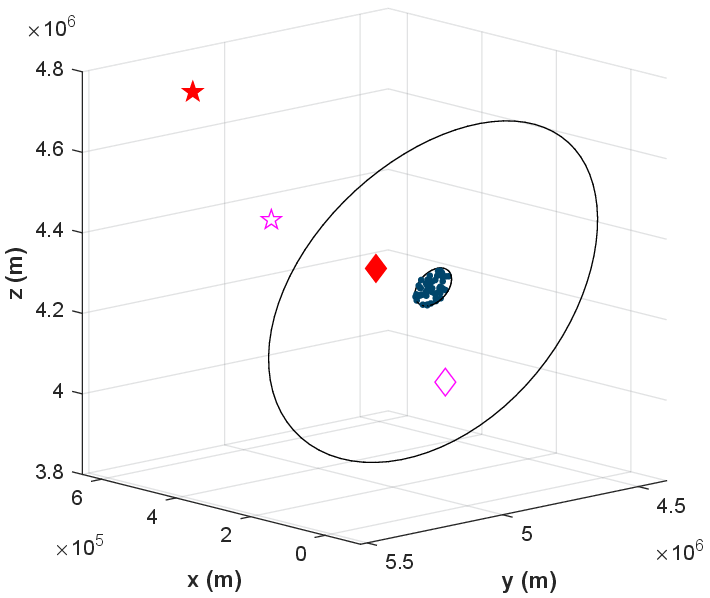}}
	\subfigure[]{
		\label{fig:des_UEdistribution5cluster}
		\includegraphics[width=0.8\linewidth]{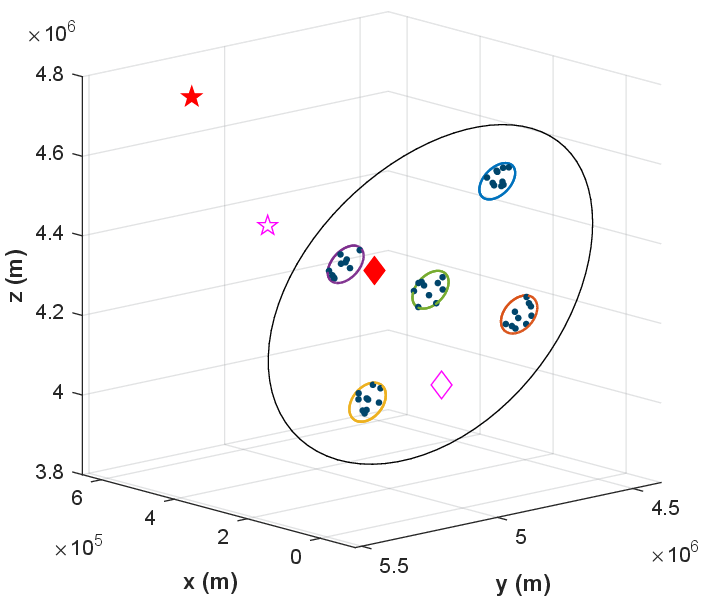}}
	\caption{The network topology at the 50-$th$ time slot under different user distributions. The points represent user terminals, the pentagrams represent satellites, and the diamonds represent sub-satellite points. (a): a high user density scenario with all users located in a circular area with a 50 $km$ radius; (b): a sparse user cluster scenario with 5 clusters of users, each cluster comprising 10 users located in a circular area with a 50 $km$ radius.}
	\label{fig:des_UEdistribution}
	\vspace{-0.1in}
\end{figure}

\begin{figure}[t]
	\centering
	\subfigcapskip=-5pt
	\subfigure[]{
		\label{fig:GP50kmUE}
		\includegraphics[width=0.95\linewidth]{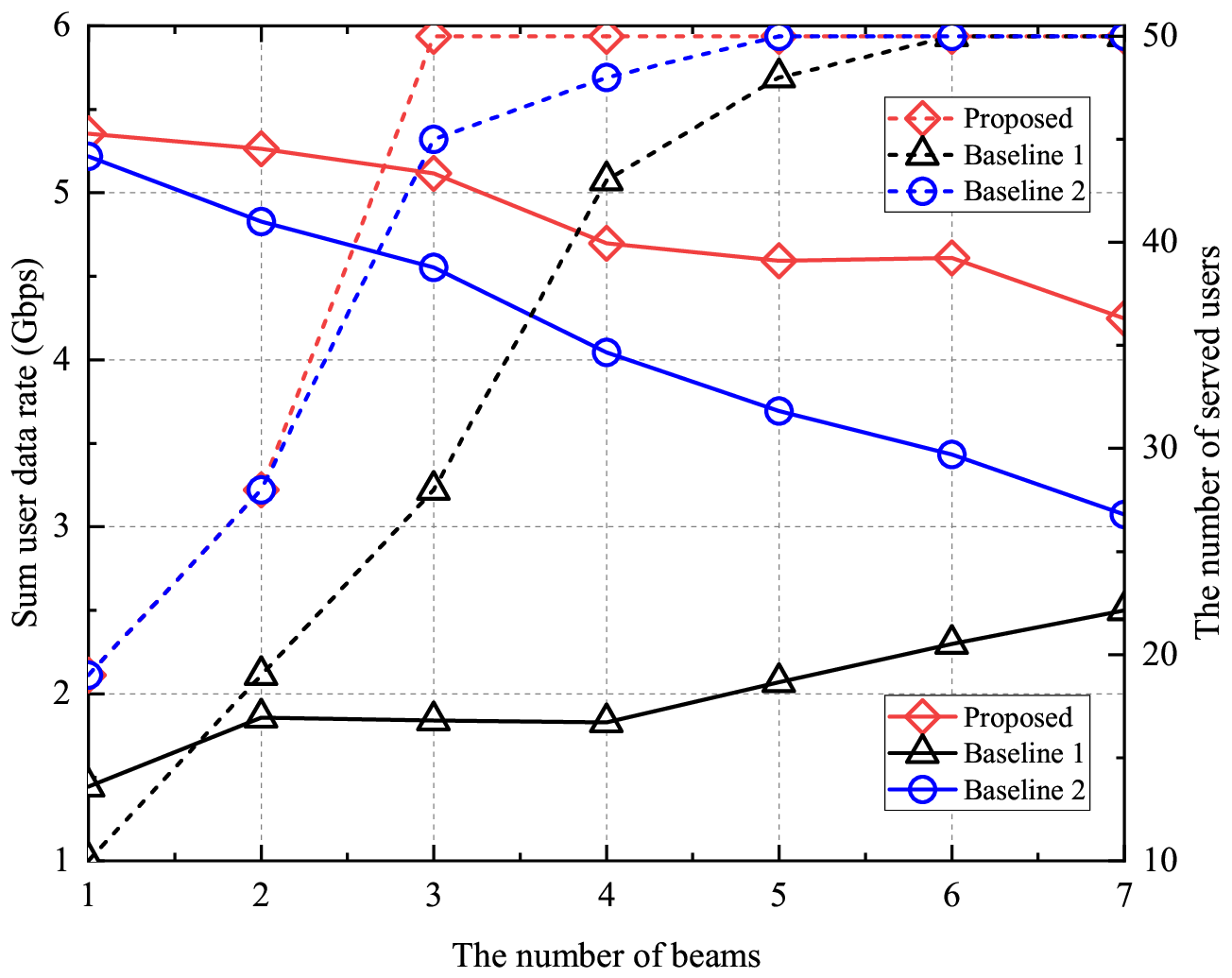}}
	\subfigure[]{
		\label{fig:GP5CUE}
		\includegraphics[width=0.95\linewidth]{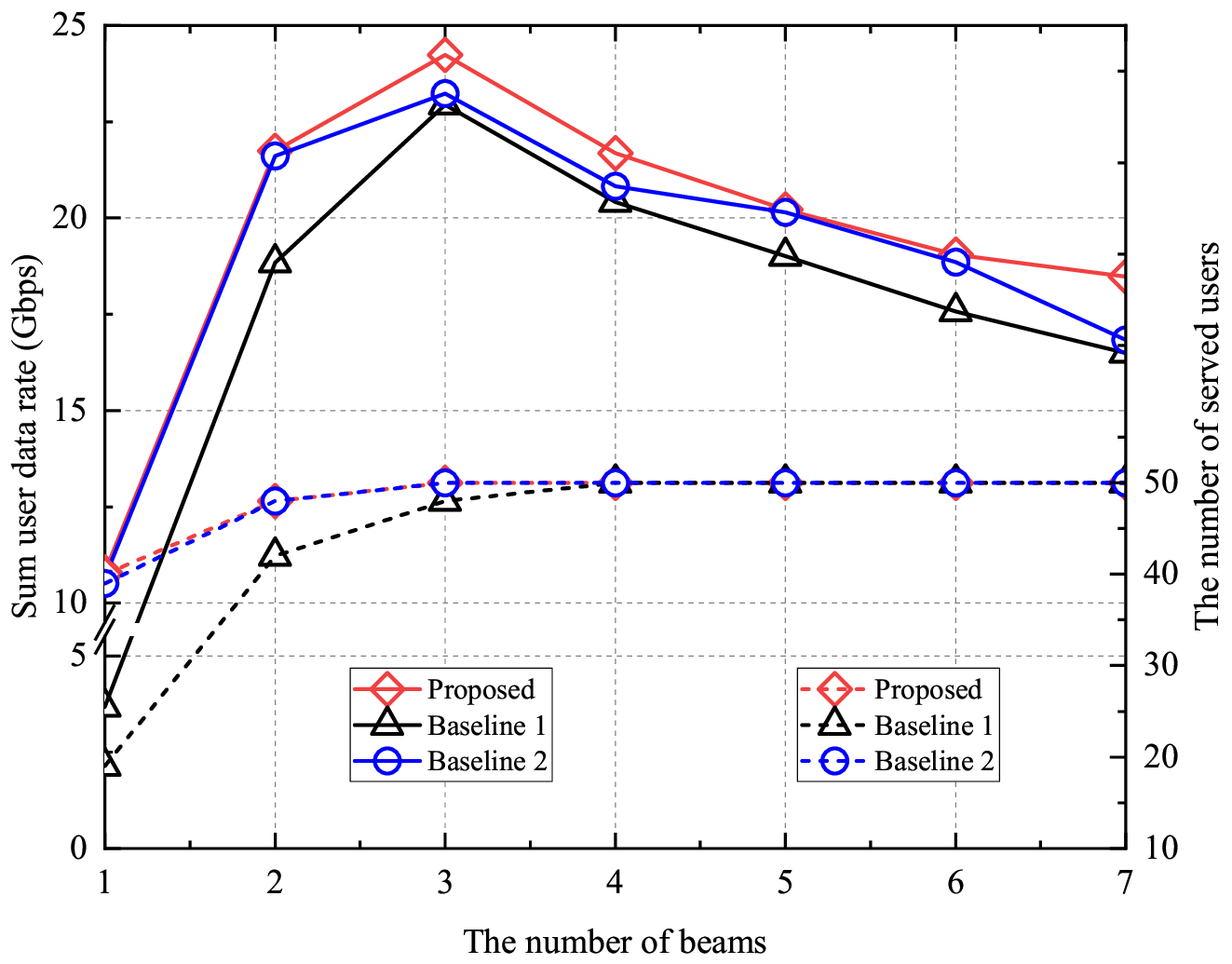}}
	\caption{Impact of the number of beams configured at each satellite on the network performance ($K=20, K^{thr}=6$). The solid lines correspond to the sum user data rate while the dash lines correspond to the number of severed users. (a): high user density scenario shown in Fig. \ref{fig:des_UEdistribution50km}; (b): sparse user cluster scenario shown in Fig. \ref{fig:des_UEdistribution5cluster}.}
	\label{fig:GPUEdistribution}
	\vspace{-0.1in}
\end{figure}

\subsection{Performance under Different User Distributions}

To further demonstrate the generality and robustness of our proposal, we investigate the performance under different user distributions.
We design two representative user distribution scenarios: one characterized by high user density, where all users are situated within a circular area with a 50 $km$ radius in the target area, and the other involving five sparsely distributed user clusters, where each cluster comprising 10 users located in a circular area with a 50 $km$ radius.
These scenarios are illustrated in Fig. \ref{fig:des_UEdistribution50km} and Fig. \ref{fig:des_UEdistribution5cluster}, respectively.
Compared to the user distribution in Fig. \ref{fig:des_UEdistributionRandom}, the distributions in Fig. \ref{fig:des_UEdistribution} show distinct user density characteristics.
As shown in Fig. \ref{fig:GPUEdistribution}, it is evident that our proposal offers significant advantages over the two baselines in terms of improving both sum user data rates and the number of served users.
For instance, in the high user density scenario with each satellite configured with $L = 7$ beams, our proposal achieves a 68\% improvement in sum user data rate compared to Baseline 1 and a 40\% improvement compared to Baseline 2.
Notably, in the high user density scenario, we observed a decreasing trend in the sum user data rate with an increase in the number of beams configured at each satellite.
This decline is attributed to stronger inter-beam interference resulting from the high user density as the number of satellite beams increases.
On the other hand, in the sparse user cluster scenario, our proposal exhibits a trend of initially increasing and subsequently decreasing sum user data rate with an increase in the number of beams.
This is due to insufficient coverage of each user cluster by a satellite beam when the total number of satellite beams is too low ($L < 3$), resulting in lower user coverage.
Conversely, when the total number of satellite beams exceeds a certain threshold, such as the number of user clusters ($L \ge 3$), it introduces significant inter-beam interference and decreases the sum user data rate.

\section{Conclusion}
\label{sec:conclusion}

In this paper, we have focused on a downlink scenario in dynamic multi-beam multi-satellite communication networks and proposed a joint optimization approach for beam direction and radio resource allocation in terms of beam direction control, subchannel assignment, and beam power allocation.
The main objective is to enhance the long-term sum user data rate while improving user fairness under onboard resource constraints.
To tackle the problem, we have decoupled it into three subproblems, namely beam direction control and time slot allocation, user subchannel assignment, and beam power allocation.
The first two subproblems are both modeled as two-sided many-to-one matching problems with externalities.
Subsequently, two swap/negotiation operation based matching algorithms have been designed, and their stability, convergence, and complexity have been analyzed.
Since the beam power allocation subproblem remains non-convex, we have exploited the successive convex approximation to solve it.
Furthermore, we have conducted extensive simulations to demonstrate the advantages of our proposal, which has increased the number of served users by up to two times and sum user data rate by up to 68\%, compared to different baselines.

\appendices

\ifCLASSOPTIONcaptionsoff
	\newpage
\fi

\bibliographystyle{IEEEtran}
\bibliography{bibUsedinPaper,bstControlForIEEEtran}

\vspace{-0.5in}
\begin{IEEEbiography}[{\includegraphics[width=1in,height=1.25in,clip,keepaspectratio]{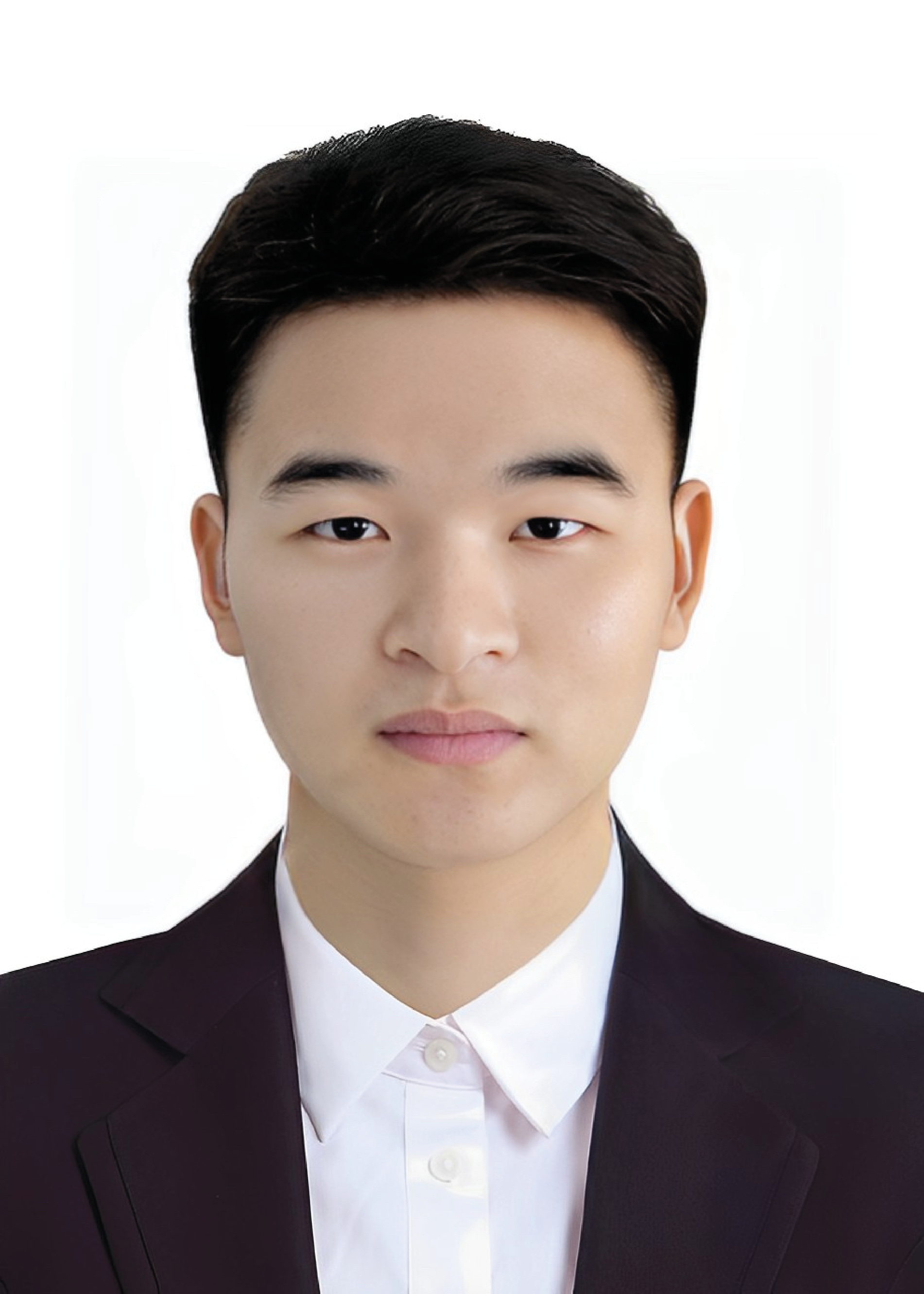}}]
	{Shuo Yuan} (Member, IEEE) received the B.S. degree from Nanchang University, Nanchang, China, and the M.E. degree in information and communication engineering from Beijing University of Posts and Telecommunications (BUPT), Beijing, China, in 2016 and 2019, respectively.
	He is currently working toward a Ph.D. degree in the State Key Laboratory of Networking and Switching Technology, BUPT, Beijing, China.
	His research interests include multi-access edge computing, intelligent computing, and LEO satellite communication.
\end{IEEEbiography}
\vspace{-0.5in}

\begin{IEEEbiography}[{\includegraphics[width=1in,height=1.25in,clip,keepaspectratio]{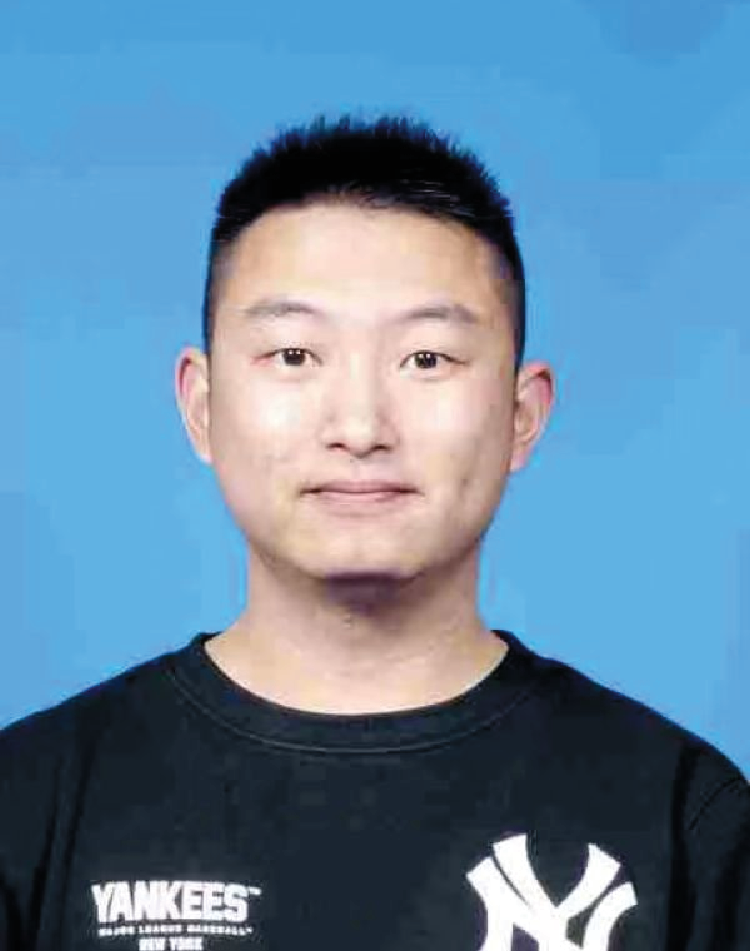}}]
	{Yaohua Sun} received the bachelor's degree (Hons.) in telecommunications engineering (with management) and the Ph.D. degree in communication engineering from Beijing University of Posts and Telecommunications (BUPT), Beijing, China, in 2014 and 2019, respectively.
	He is currently an Associate Professor with the School of Information and Communication Engineering, BUPT.
	His research interests include intelligent radio access networks and LEO satellite communication.
	He has published over 30 papers including 3 ESI highly cited papers. He has been a Reviewer for IEEE \textsc{Transactions on Communications} and IEEE \textsc{Transactions on Mobile Computing}.
\end{IEEEbiography}
\vspace{-0.5in}

\begin{IEEEbiography}[{\includegraphics[width=1in,height=1.25in,clip,keepaspectratio]{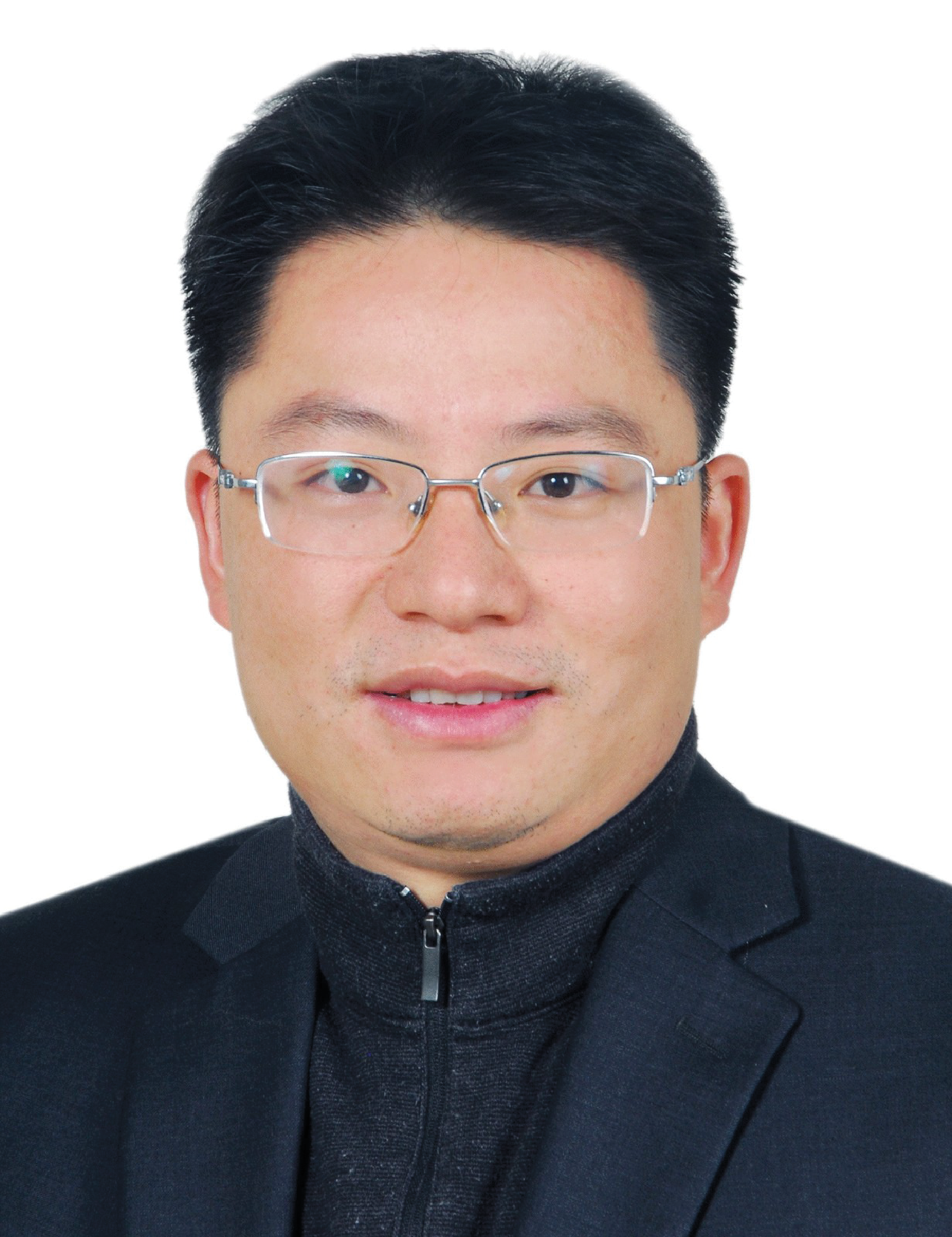}}]
	{Mugen Peng} (Fellow, IEEE) received the Ph.D. degree in communication and information systems from the Beijing University of Posts and Telecommunications, Beijing, China, in 2005. In 2014, he was an Academic Visiting Fellow at Princeton University, Princeton, NJ, USA.
	He joined BUPT, where he has been the Dean of the School of Information and Communication Engineering since June 2020, and the Deputy Director of the State Key Laboratory of Networking and Switching Technology since October 2018.
	He leads a Research Group focusing on wireless transmission and networking technologies with the State Key Laboratory of Networking and Switching Technology, BUPT.
	His main research interests include wireless communication theory, radio signal processing, cooperative communication, self-organization networking, non-terrestrial networks, and Internet of Things.
	He was a recipient of the 2018 Heinrich Hertz Prize Paper Award, the 2014 IEEE ComSoc AP Outstanding Young Researcher Award, and the Best Paper Award in IEEE ICC 2022, JCN 2016, and IEEE WCNC 2015.
	He is/was on the Editorial or Associate Editorial Board of IEEE \textsc{Communications Magazine}, IEEE \textsc{Network Magazine}, IEEE \textsc{Internet of Things Journal}, IEEE \textsc{Transactions on Vehicular Technology}, and IEEE \textsc{Transactions on Network Science and Engineering}.
\end{IEEEbiography}
\vspace{-0.5in}

\begin{IEEEbiography}[{\includegraphics[width=1in,height=1.25in,clip,keepaspectratio]{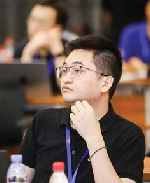}}]
	{Renzhi Yuan} (Member, IEEE) received the B.S. degree in mechanical engineering from Tongji University, Shanghai, China, in 2011, the M.S. degree in precision instrument from Tsinghua University, Beijing, China, in 2017, and the Ph.D. degree in electrical engineering from the University of British Columbia, Kelowna, BC, Canada, in 2021. He is currently a Distinguished Research Fellow with the School of Information and Communication Engineering and the State Key Laboratory of Networking and Switching Technology, Beijing University of Posts and Telecommunications. His research interests include the space-air-ground integrated communications, wireless optical communications, and quantum communications.
\end{IEEEbiography}

\end{document}